\documentclass[a4paper]{article}
\usepackage{fullpage}

\usepackage[T1]{fontenc}
\usepackage{amsfonts,amsmath,amssymb,amsthm}
\usepackage{tcolorbox}
\usepackage{multicol}
\usepackage{environ}
\usepackage[symbol]{footmisc}

\newtheorem{theorem}{Theorem}
\newtheorem{lemma}{Lemma}
\newtheorem{fact}{Fact}

\newtheorem{proposition}{Proposition}
\newtheorem{conjecture}{Conjecture}

\usepackage{xcolor}
\usepackage{algorithm,algorithmicx}
\usepackage[noend]{algpseudocode}
\MakeRobust{\Call}

\algblockdefx{Nest}{EndNest}{}{}
\algtext*{Nest}{}
\algtext*{EndNest}{}

\definecolor{light-gray}{gray}{0.95} 
\usepackage{listings}
\lstdefinestyle{mystyle}{
    basicstyle=\ttfamily\small, 
    commentstyle=\color{gray},
    backgroundcolor=\color{light-gray},
    morecomment=[f][\color{grey}][0]{\#},
    morecomment=[f][\color{grey}][0]{//},
    numbers=left,                    
    numbersep=5pt,    
    mathescape=true,
    keywordstyle=\textbf,
    morekeywords={while,if,then,else,true,false,do,:,parallel,in,def,until},
}
\newcommand{\transitionrelComment}[7][<]{& (#2) \; && #3 + #4 &&\stackrel{#1}{\longrightarrow} &&#5 + #6 \;
  &\triangleright\  #7\\}

\newcommand{\transitionrel}[6][<]{& (#2) \; && #3 + #4 &&\stackrel{#1}{\longrightarrow} &&#5 + #6  &\;\\}

\newcommand{\transitionrelnull}[6][<]{& (#2) \; && #3 + #4 &&\stackrel{#1}{\longrightarrow} &&#5 \  #6  &\;\\}

\newcommand{\transition}[5]{\transitionrel[]{#1}{#2}{#3}{#4}{#5}}

\newcommand{\transitionnull}[5]{\transitionrelnull[]{#1}{#2}{#3}{#4}{#5}}

\newcommand{\transitionreltwo}[6][]{& #3 + #4 &&\stackrel{#1}{\longrightarrow} &&#5 + #6  &\;\\}

\newcommand{\transitionrelnulltwo}[6][]{&  #3 + #4 &&\stackrel{#1}{\longrightarrow} &&#5 \  #6  &\;\\}

\AtBeginDocument{%
  \providecommand\BibTeX{{%
    \normalfont B\kern-0.5em{\scshape i\kern-0.25em b}\kern-0.8em\TeX}}}

\title{Selective Population Protocols}

\author{Adam Ga{\'n}czorz\\
Wroc\l aw University
\and Leszek G{\k a}sieniec\\
University of Liverpool
\and Tomasz Jurdzi{\'n}ski\\ 
Wroc\l aw University
\and Jakub Kowalski\\
Wroc\l aw University
\and Grzegorz Stachowiak\\
Wroc\l aw University}

\begin{document}

\maketitle

\begin{abstract}
The model of population protocols provides a universal platform to study distributed processes driven by pairwise interactions of anonymous agents. While population protocols present an elegant and robust model for randomized distributed computation, their efficiency wanes when tackling issues that require more focused communication or the execution of multiple processes. To address this issue, we propose a new, selective variant of population protocols by introducing a partition of the state space and the corresponding conditional selection of responders. We demonstrate on several examples that the new model offers a natural environment, complete with tools and a high-level description, to facilitate more efficient solutions.  

In particular, we provide fixed-state stable and efficient solutions to two central problems: leader election and majority computation, both with confirmation. This constitutes a separation result, as achieving stable and efficient majority computation requires $\Omega(\log n)$ states in standard population protocols, even when the leader is already determined. Additionally, we explore the computation of the median using the comparison model, where the operational state space of agents is fixed, and the transition function determines the order between (arbitrarily large) hidden keys associated with interacting agents. Our findings reveal that the computation of the median of $n$ numbers requires $\Omega(n)$ time. Moreover, we demonstrate that the problem can be solved in $O(n\log n)$ time, both in expectation and with high probability, in standard population protocols. In contrast, we establish that a feasible solution in selective population protocols can be achieved in $O(\log^4 n)$ time.
\end{abstract}


\section{Introduction}\label{s1}

\setlength\columnsep{14pt}

The standard model of population protocols originates from the seminal paper \cite{DBLP:conf/podc/AngluinADFP04}, providing tools suitable for the formal analysis of {\em pairwise interactions} between simple, indistinguishable entities referred to as {\em agents}. These agents are equipped with limited storage, communication, and computation capabilities.
When two agents engage in a direct interaction, their states change according to the predefined {\em transition function}, which is an integral part of the population protocol. The weakest possible assumptions in population protocols pertain to the fixed (constant size) operational {\em state space} of agents, and the size of the population $n$ is neither known to the agents nor hard-coded in the transition function.
It is assumed that a protocol starts in the predefined {\em initial configuration} of agents’ states representing the input, and it {\em stabilizes} in one of the {\em final configurations} of states representing the solution to the considered problem.
In the {\em probabilistic variant} of population protocols adopted here, in each step of a protocol, the {\em random scheduler} selects an ordered pair of agents: the {\em initiator} and the {\em responder}, which are drawn from the whole population uniformly at random. The lack of symmetry in this pair is a powerful source of random bits utilized by population protocols.
In the probabilistic variant, in addition to efficient {\em state utilization}, one is also interested in the {\em time complexity}, where the {\em sequential time} refers to the number of interactions leading to the stabilization of a protocol in a final configuration. More recently, the focus has shifted to the {\em parallel time}, or simply the {\em time}, defined as the sequential time divided by the size nn of the whole population. The (parallel) time reflects on the parallelism of simultaneous independent interactions of agents utilized in {\em efficient population protocols} that stabilize in time 
$O(\mbox{poly} \log n)$.
All protocols presented in this paper are {\em stable} (always correct) and guarantee stabilization time with high probability (whp) defined as $1-n^{-\eta},$ 
for a constant $\eta>0.$

There are already several efficient protocols known for solving central problems in distributed computing, including {\em leader election}~\cite{DBLP:conf/soda/AlistarhAEGR17,GasieniecS21,BerenbrinkGK20}, {\em majority computation}~\cite{DBLP:conf/soda/AlistarhAG18,DotyEGSUS21}, and the {\em plurality problem}~\cite{BankhamerBBEHKK22}. While these protocols are efficient in terms of time, they rely on non-constant state space utilization and operate indefinitely. That is, they are not able to declare stabilization with probability 1. Moreover, the most efficient protocols are often non-trivial and hard to analyze.
One can circumvent some of these deficiencies by relaxing probabilistic expectations, e.g., by dropping assumptions about the necessity of stabilization in protocols with predefined input \cite{DBLP:conf/podc/KosowskiU18, DBLP:journals/corr/abs-1802-06872}, as well as in self-stabilizing protocols~\cite{DBLP:conf/podc/BurmanCCDNSX21}. While such relaxation provides certain benefits, it does not solve some major deficiencies of the standard model, including depleting in time the number of meaningful interactions, limited computational power, and inefficient space-time trade-offs.

In order to circumvent some of these deficiencies,   
we propose a new {\em selective} variant of population protocols by imposing a simple {\em group (partition) structure} on the state space together with a conditional choice of the responder during random interacting pair selection.
This model provides a natural extension of {\em passive mobile} sensor networks adopted in~\cite{DBLP:conf/podc/AngluinADFP04}, where the focus is on single channel communication.
In the new model the agents communicate over multiple communication channels, where each channel corresponds to one of a fixed number of groups (partitions) of the state space. Specifically, only agents currently listening on some specific communication channel $C$ (their states belong to the corresponding group of states $\mathcal{G}_C$) are able to receive and respond to messages transmitted over this channel by agents with state indicating $\mathcal{G}_C$ as the {\em target group}.    
Alternative models with biased communication were previously used in the context of stochastic chemical reaction networks in~\cite{SoloveichikCWB08} and data collection with non-uniform schedulers in~\cite{DBLP:conf/podc/BurmanCCDNSX21}. The adopted selective model also refers to biased choices in nature studied earlier, in the context of small-world phenomena, where closer location in space results in a more likely interaction~\cite{Klein00}, or social preference, where agents with a greater array of similar attributes are more likely to know one another and, in turn, to interact~\cite{Ken98}.
A different motivation to study selective population protocols refers to more structural variant of population protocols known as {\em network constructors}, in which agents are allowed to be connected. As the expected parallel time to manipulate a specific edge is $\Theta(n),$ see, e.g.,~\cite{MS17,DBLP:conf/stacs/GasieniecSS23}, 
the design of truly efficient protocols in this model is not currently feasible. 
Utilizing the concept of selective population protocols, one can give greater probabilistic bias to interactions along already existing edges, and in turn enable more efficient computation, comparable with another variant known as graphical population protocols~\cite{DBLP:conf/opodis/AlistarhGR21, DBLP:conf/podc/AlistarhRV22}.  

\subsection{Our results}

In this paper, we present initial studies on the (parallel) efficiency and stability of selective population protocols. We begin by discussing fundamental properties of this new variant, introducing the notion of {\em fragmented parallel time} as an equivalent measure to parallel time in the standard population protocol model. Additionally, we highlight that selective protocols offer a natural mechanism for deterministic emptiness (zero) testing.
It is known, as indicated in \cite{DBLP:journals/dc/AngluinAE08a}, that such a test enables efficient simulation of $O(\log n)$-space Turing Machines with high probability. In contrast, we highlight that such simulations in selective protocols are not only efficient but also stable. Selective protocols can be utilized to design algorithms within this class that are both efficient and stable.
Furthermore, we present fixed-state efficient and stable solutions to two central problems: leader election and majority computation (with confirmation, i.e., all agents stabilize while being aware of the conclusion of the process). This result is noteworthy as stable efficient majority computation requires $\Omega(\log n)$ states in standard population protocols \cite{DBLP:conf/soda/AlistarhAG18}, even when the leader is given.
%
%
%
We also introduce the first non-trivial study on median computation in population protocols. We adopt a comparison model in which the operational state space of agents is constant, and the transition function determines the order between (arbitrarily large) hidden keys associated with the interacting agents. We demonstrate that computing the median of $n$ numbers requires $\Omega(n)$ parallel time and the problem can be solved in $O(n\log n)$ parallel time in expectation and with high probability (whp) in standard population protocols. In contrast, we present an efficient median computation in selective population protocols, achieving $O(\log^4 n)$ parallel time.
Furthermore, we delve into the appropriateness of selective protocols for the high-level design of algorithmic solutions.

\section{Selective Population Protocols}\label{s2}

As discussed in Section~\ref{s1}, in the standard population protocol model, the random scheduler draws consecutive pairs of interacting agents uniformly at random from the entire population. This is done irrespective of whether the states of interacting agents match some rule of the transition function or not. Consequently, many random pairwise interactions do not result in a transition and, in turn, do not bring the population closer to a final configuration.

In {\em selective population protocols} random selection of interacting pairs is 
realised differently.
In particular, the fixed-state space of agents
$\mathcal{S}$ is partitioned into $l$ 
groups of states
$\mathcal{G}_1,\mathcal{G}_2,\ldots,\mathcal{G}_l.$ 
In addition, 
any state $s\in \mathcal{G}_j,$ 
is mapped onto its {\em target group} $\mathcal{G}_{i(s)}.$ 
We say that an interaction is {\em internal} if $s\in {\mathcal G}_{i(s)},$ and is {\em external} otherwise. 
%
This mapping is used during an attempt to form a random pair of interacting agents.
The random scheduler first draws the initiator uniformly from the whole population. 
This is followed by drawing the responder uniformly from all agents (different to the initiator) currently residing in states belonging to target group $\mathcal{G}_{i(s)}$,
where $s$ is the state of the initiator. 
%
%
The rules of the transition function manage two different outcomes of interaction attempts:

\begin{enumerate}
\item {\em Biased communication} 
\begin{description}
    \item[Meaningful interaction] (successful biased interaction attempt) 
    \nopagebreak\\
    $s+\mathcal{G}_{i(s)}|t\rightarrow s'+t'.$
\end{description}
The purpose of meaningful interactions is to advance and in turn to maintain efficiency of the computing process.\\

\item {\em Interaction availability test}
\begin{description}    
    \item[Emptiness (zero) test] (unsuccessful external interaction attempt) 
    \nopagebreak\\
    $s+\mathcal{G}_{i(s)}|null\rightarrow s'.$
    \item[Singleton test] (unsuccessful internal interaction attempt)
    \nopagebreak\\
    $s+\mathcal{G}_{i(s)}|null\rightarrow s'.$
\end{description}
The two tests are mainly used to determine completion of computation processes.
\end{enumerate}
 


\noindent
{\bf One-way epidemic} Consider a  communication primitive known as {\em one-way epidemic}~\cite{DBLP:journals/dc/AngluinAE08a} in which state $1$ of the source agent is propagated to all other agents initially being in state $0.$ The transition function has only one rule in the standard population protocol model
$$ 1 + 0 \rightarrow 1 + 1. $$
It is known that such epidemic process is stable and efficient, i.e., one-way epidemic stabilises with the correct answer in parallel time $O(\log n)$ whp, however 
during the final stages of the epidemic process the expected fraction of {\em meaningful interactions} 
decreases dramatically. Assume, that the state space $S=\{0,1\}$ is partitioned into two singleton groups $\mathcal{G}_0=\{0\}$ and $\mathcal{G}_1=\{1\}$ in the new variant, and
we have two transition rules instead:

\begin{small}
\begin{tcolorbox}
\begin{multicols}{2}
\noindent
\begin{alignat*}{5}
    \transition{1}{1}{\mathcal{G}_0|0}{1}{1}
\end{alignat*}
\begin{alignat*}{5}
    \transitionnull{2}{1}{\mathcal{G}_0|null}{Stop}{}
\end{alignat*}
\end{multicols}
\vspace{-3em}
\end{tcolorbox}
\end{small}
%

Now every interaction initiated by an agent in state $1$ is either meaningful, when there are still uninformed agents, or changes the initiator's state to $Stop,$ which indicates the end of the epidemic process, and in turn the next stage of computation not requiring group ${\mathcal G}_0$.

\subsection{Beyond Presburger Arithmetic}
\label{Beyond}


%

We find in~\cite{DBLP:journals/dc/AngluinAE08a} that the emptiness test, also known as {\em zero test}, is a powerful tool enabling efficient simulation of $O(\log n)$-space Turing Machine. 
The two-stage randomised simulation from~\cite{DBLP:journals/dc/AngluinAE08a} is based on simulation of {\em Register Machines} known to be equivalent with $O(\log n)$-space Turing Machines~\cite{minsky_67}. 
This approach hinges on the presence of a unique leader, crucial for achieving efficient and stable computations, which efficient computation requires at least $\Omega(\log\log n)$ states~\cite{DBLP:conf/soda/AlistarhAEGR17}.
%
An alternative randomized two-stage simulation of Turing Machines, detailed in~\cite{SoloveichikCWB08} within the context of a related {\em stochastic chemical reaction network} model, utilises the concept of {\em clockwise Turing Machines}~\cite{NearyW09}. Both simulations rely on zero tests, the correctness of which can be assured only with high probability in the adopted models, rendering them unsuitable for deployment in stable protocols.
%
%
In~contrast, selective protocols equipped with deterministic emptiness test provide a suitable platform for the design of efficient and stable {\em fixed-state} solutions in $O(\log n)$-space complexity class.

As the primary focus of this paper centers on the (parallel) {\em efficiency} of selective population protocols, and the computational power of such protocols is inherited from standard population protocols, we direct the reader to~\cite{DBLP:journals/dc/AngluinAE08a} for full simulation details. Instead, our current study delves into the parallelism of selective protocols, presenting several separation results. This includes {\em majority computation}, where any efficient stable algorithm in standard protocols requires $\Omega(\log n)$ states while a fixed-state space allows for an $O(\log n)$-time stable solution, as demonstrated in Section \ref{Majority}.

Several efficient algorithms presented in this paper follow a more direct approach, relying on a single stabilization process. Examples include the efficient and stable leader election and majority computation discussed in Sections~\ref{LE} and~\ref{Majority}, respectively.
However, in more complex solutions, the need arises for a leader to act as the {\em "program counter,"} overseeing the proper execution of potentially numerous individual stabilization processes encoded in the transition function in the correct order. This encompasses the preparation of input for each individual process, ensuring its proper termination, and further interpreting the output. It's worth noting that, due to state partitioning in selective protocols, each individual stabilization process can be executed on a distinct partition of states. This allows several processes to run efficiently at the same time, as demonstrated in the efficient ranking protocol presented in Section~\ref{final}, where multiple leaders are employed.
The leader is also responsible for translating the output from one process to the input of its successor. This is achieved by rewriting states (from one partition to another) via one-way epidemic. Ultimately, the termination of any process, including rewriting, is recognized through either an emptiness or singleton test.

As a warm up, we present here an example of (inefficient) leader based multiplication protocol. In this problem, the input is formed of two subpopulations of agents $X$ and $Y,$ and the task is to create a new
collection $Z$ of size $|X|\cdot |Y|.$ Please note that an efficient multiplication protocol can be found in the Appendix, in Section~\ref{FM}.

\subsubsection{Example of Leader Driven Selective Protocol}\label{SM}

%
%
%
%
%
%
%
%
During execution of the multiplication protocol, the leader deletes agents from $X$ one by one, where each agent deletion coincides with enlargement of $Z$ by $|Y|$ elements.
Note that for simplicity of the argument our multiplication protocols provide construction of set $Z$ rather than solving the decision problem $|X|\cdot |Y|\le |Z|.$ 
In a solution to the decision problem one would remove (one by one) agents from set $Z$ rather than adding them to it. Nevertheless, the logic and the complexity of both solutions are exactly the same.

The {\em state space}
$\mathcal{S}=\{L_{in}, L_Y, L_{\hat y}, L_Z, L_{out}, X, Y, \hat{Y}, Z, free\}$ of the multiplication protocol is 
partitioned into five target groups including singletons 
$$G_X=\{X\},
G_Y=\{Y\}, 
G_{\hat Y}=\{\hat Y\}, 
G_Z=\{Z\}, 
$$
were groups $G_X,G_Y$ and $G_Z$ coincide with collections of agents $X,Y$ and $Z$ respectively, and $G_{\hat Y}$ contains temporarily marked agents from $Y.$
The remaining agents, i.e., the leader and all free (unused) agents have their states in
$$G_R=\{L_{in},L_Y,L_Z,L_{\hat Y},L_{out},Free\},$$ 
where the leader enters the protocol in state $L_{in}$ and 
concludes in state  $L_{out}$. 
State $L_{Y}$ is used to mark an agent in state $Y$ (already used in the current round of multiplication) as $\hat Y$ and state $L_{Z}$ to create a new agent in state $Z.$
Finally
$L_{\hat Y}$ is used to unmark all agents in state $\hat Y$ by reverting their state to $Y.$

The {\em transition function} of the protocol has seven rules
\begin{small}
\begin{tcolorbox}
\begin{multicols}{2}
    \noindent 
    \begin{alignat*}{5}
    \transition{1}{L_{in}}{\mathcal{G}_X|X}{L_Y}{\text{Free}}
    \transitionnull{2}{L_{in}}{\mathcal{G}_X|null}{L_{out}}{}
    \transition{3}{L_Y}{\mathcal{G}_Y|Y}{L_Z}{\hat{Y}}
    \transitionnull{4}{L_Y}{\mathcal{G}_Y|null}{L_{\hat{Y}}}{}
    \end{alignat*}
    \begin{alignat*}{5}
    \transition{5}{L_Z}{\mathcal{G}_R|\text{Free}}{L_Y}{Z}
    \transition{6}{L_{\hat{Y}}}{\mathcal{G}_{\hat{Y}}|\hat{Y}}{L_{\hat{Y}}}{Y}
    \transitionnull{7}{L_{\hat{Y}}}{\mathcal{G}_{\hat{Y}} | null}{L_{in}}{}
    \end{alignat*}
\end{multicols}
\vspace{-3em}
\end{tcolorbox}
\end{small}

Rule (1) deletes one agent in $X$ and prompts the leader now in state $L_Y$ to mark one agent in $Y$ as $\hat Y$ using rule (3).
Subsequently, the leader now in state $L_Z$ inserts one agent to $Z$ reverting its state to $L_Y$ using rule (5). 
The alternating process of marking $Y$'s and forming new $Z$'s concludes when all agents from $G_Y$ are moved to $G_{\hat Y}.$  
This is when rule (4) instructs the leader now in state $L_{\hat Y}$ to unmark all $\hat Y$'s, i.e., moving all agents in $G_{\hat Y}$ back to $G_Y$ using rule (6).
And when unmarking is finished the leader state is changed to $L_{in}$ by rule (7) to make another attempt to delete an agent from $X.$ And when this is no longer possible rule (2) concludes the multiplication protocol by changing the state of the leader to $L_{out}.$

\subsection{Parallelism of Selective Protocols}
\label{Par}

%
Recall that in population protocols, the (parallel) time of a sequence $I$ of interactions is defined as $|I|/n$. This definition is motivated by the observation that in a sequence of $xn$ interactions, each agent has, on average, $x$ interactions. However, it is noteworthy that only for $x=\Omega(\log n)$ does each agent engage in $\Theta(x)$ interactions whp.
An interesting finding presented in \cite{DBLP:journals/ipl/CzumajL23} demonstrates that in this latter case, the sequence $I$ can be simulated in time $\Theta(x)$ on a parallel computer whp.
In the new variant, where the choice of the responder is likely biased, we must adopt a more nuanced definition of parallelism.

\subsubsection{Fragmented  Parallel Time}
In the novel selective variant of population protocols, the initiator is uniformly chosen at random. Consequently, in a sufficiently long sequence of interactions, any agent serves as the initiator with the same frequency, aligning with the pattern observed in the standard model. This stands in contrast to the selection of responders, where certain agents are more likely to be chosen than others.

\noindent {\bf Example.}  Consider an epidemic process with the state space $S=\{0,1,1^*\}$ partitioned into groups: $\mathcal{G}_0=\{0\}$ with uninformed agents, $\mathcal{G}_1=\{1\}$ with active informers, and $\mathcal{G}_{1^*}=\{1^*\}$ with informed and already rested agents, governed by two transition rules:
\begin{small}
\begin{tcolorbox}
\begin{multicols}{2}
\noindent
\begin{alignat*}{5}
    \transition{1}{0}{\mathcal{G}_1|1}{1^*}{1}
\end{alignat*}
\begin{alignat*}{5}
    \transitionnull{2}{1}{\mathcal{G}_0|null}{1^*}{}
\end{alignat*}
\end{multicols}
\vspace{-3em}
\end{tcolorbox}
\end{small}
%
%
If in the initial configuration there is exactly one informed agent in state $1,$ all other agents in state $0$ contact this agent to get informed and rest, see rule (1). In the last meaningful
 interaction rule (2) rests the unique informer. While the number of interactions before stabilisation with all agents resting in state $1^*$ is $O(n\log n)$ whp, the parallelism of this epidemic process is very poor as only one agent informs others as the responder.

%
This potential imbalance in the workload of individual agents can be captured by tracking the frequency at which agents act as responders.
To handle this imbalance, we propose a more subtle definition of (parallel) time.
This new definition is whp asymptotically equivalent to the definition and the properties of time used in the standard model, see Lemma~\ref{subtle}.

Consider ways to divide the sequence of interactions $I$
into subsequent disjoint {\em chunks}, where each chunk is a sequence 
of consecutive interactions in which any agent has at most $10\ln n$ interactions as the  responder.
If the minimum number of chunks for such divisions is $k$, then
we say that the {\em fragmented parallel time}, or in short the {\em fragmented time}, is
$T_F=k\ln n$.

As defined in Section~\ref{s1}, $\eta$ is the quality parameter in the definition of high probability.

\begin{lemma}\label{subtle}
Consider a sequence of interactions $I$ in the standard population protocol model executed in time $T=|I|/n.$ 
If the fragmented time $T_F=k\ln n,$
for $k>11\eta/35$ and large enough $n$,
then
$T/10\le T_F \le 2T$
whp.
\end{lemma}

\begin{proof}
The total number of interactions during fragmented time $T_F$ does not exceed $10kn\ln n$, so we always have $T/10\le T_F$.
It remains to show that $T_F\le 2T$ whp.

Let us first estimate the probability that a given chunk
corresponds to time smaller than $\ln n$.
This probability is not greater than the
probability that in time $\ln n$ (starting at the beginning of the chunk) some
agent has the responder type interactions greater than $10\ln n$.
By Chernoff bound\footnote[4]{
We utilise the Chernoff bound variant:
$\Pr\left(X>(1+\delta)\mathbb{E}X\right)<
\exp\left(-\delta^2 \mathbb{E}X/(1+\delta)\right)$ for $\delta>0$.
}, the probability that in time $\ln n$ a given agent experiences $X>10\ln n=10\mathbb{E}X$ interactions as the responder can be estimated by
\begin{align*}
\Pr(X>10\ln n)=&\Pr\left(X>(1+9)\mathbb{E}X\right)\\ <& \exp\left(-\frac{9^2}{2+9}\mathbb{E}X\right)
= \exp\left(-\frac{81}{11}\ln n\right) = n^{-81/11}.
\end{align*}
By the union bound the probability that in time $\ln n$ some agent interacts as the responder more
than $10\ln n$ times is smaller than $n^{-70/11}$.
Thus, for $n$ large enough and $11\eta/35<k,$
the probability that at least half of $k$ chunks correspond to time smaller than $\ln n$
does not exceed
$$\binom{k}{k/2} \left(n^{-70/11}\right)^{k/2}< 2^k n^{-35k/11}<(2^kn^{-k/5})n^{-\eta}<n^{-\eta}.$$
In turn, whp we obtain time at least $\frac{k}2\ln n$, and $T_F\le 2T$.
\end{proof}

Now we formulate a lemma allowing us to analyze fragmented time 
in the new model. This lemma will be used in the analysis of leader election and majority computation protocols.

\begin{lemma}\label{YouOnlyTakeOnceLemma}
Consider an interval of interactions $I$, s.t., $|I|>\frac{110}{35}\eta n\ln n$.
If~every agent acts as the responder in an
external interaction in $I$ at most once,
then the fragmented parallel time of $I$ is $\Theta(|I|/n)$ whp.
\end{lemma}

\begin{proof}
   The total number of chunks of the fragmented time is at least $k\ge \frac{|I|}{10n\log n}$,
so the fragmented time $T_F\ge\frac{|I|}{10n}$.
We show that the fragmented time $T_F\le \frac{2|I|}{n}$ whp.

Consider a fixed agent in  a given interaction.
We first observe that the probability of an event $A$, that an interaction is internal and this agent acts as the responder is at most $1/n$.
For an agent belonging to  a group of size 1, when this agent can be counted as both the initiator and the responder, the probability of event $A$ does not exceed $1/n$.
For an agent belonging to
a group of size $g>1$ this probability is at most $\frac{g-1}{n}\cdot\frac{1}{g-1}=\frac{1}{n}$.

Let us divide all interactions into maximal subseries such that
for any fixed agent there are at most $10\ln n-1$ events $A$ involving this agent.
Note that these subseries are simultaneously chunks
of interactions in which any agent acts as the responder
at most $10\ln n$ times since any agent is a responder in $I$ in an external interaction at most once.
Let us first estimate the probability that a given subseries has less than $n\ln n$ interactions.
This probability is not greater than the
probability that during $n\ln n$ interactions (counting from the beginning of the subseries) event
$A$ happens for some agent at least $10\ln n$ times.
Using calculations from Lemma~\ref{subtle}, one can estimate that
this probability for a specific agent is smaller than $n^{-81/11}$.

By the union bound the probability that in 
a given subseries event $A$ occurs for some agent at least $10\ln n$ times does not exceed $n^{-70/11}$.
Analogously to the proof of Lemma \ref{subtle}, for sufficiently large $n$ and $k \ge \frac{|I|}{10n\ln n} > \frac{11\eta}{35}$, the probability that at least half of the $k$ subseries correspond to times smaller than $\ln n$ is less than $n^{-\eta}$.

As this event occurs with negligible probability, we obtain 
time at least $\frac{2|I|}{n}$ whp.
\end{proof}

Recall that if a group is a singleton, an
attempt to execute pairwise interaction within this group fails. 
This is observed by the initiator via singleton test.
Note also that such failed interactions do not affect parallelism as
each failed interaction is attributed to the initiator.

The next lemma sanctions the analysis of more complex protocols utilising the leader.

\begin{lemma}\label{TakeOnceAndLeader}
Consider an interval of interactions $I$, s.t., $|I|>\frac{60}{13}\eta n\ln n$.
Assume also that every agent acts as the responder in an
external interaction, which is not initiated by the leader in $|I|$ at most once,
then the fragmented time of $I$ is $\Theta(|I|/n)$ whp.
\end{lemma}

\begin{proof}
The total number of chunks of the fragmented time is  $k\ge \frac{|I|}{10n\ln n}$,
so the fragmented time $T_F\ge\frac{|I|}{10n}$.
We show, that the fragmented time $T_F\le \frac{2|I|}{n}$ whp.

Let us consider a fixed agent in  a given interaction.
We first observe that the probability of an event $A$, that an interaction is internal or with the leader and the agent is a responder is at most $2/n$.
For an agent belonging to  a group of size 1, when this agent can be counted as the initiator and the responder, the probability of event $A$ does not exceed $2/n$.
For an agent belonging to
a group of size $g>1$ this probability is at most $\frac{g-1}{n}\cdot\frac{1}{g-1}+\frac{1}{n}=\frac{2}{n}$.

Let us divide interactions into maximal subseries such that
for any fixed agent there are at most $10\ln n-1$ events $A$ involving this agent.
Note that these subseries are simultaneously chunks
of interactions in which any agent acts as the responder
at most $10\ln n$ times since any agent is a responder in $|I|$ at most once in an external interaction with a non-leader.
Let us first estimate the probability that a given subseries has less than $n\ln n$ interactions.
This probability is not greater than the
probability that during $n\ln n$ interactions (starting at the beginning of the subseries) event
$A$ happens for some agent at least $10\ln n$ times.
By Chernoff bound, the probability that for a given agent event $A$ in $n\ln n$ interactions occurs $X\ge 10\ln n\ge 5\mathbb{E}X$ can be estimated as follows
\begin{align*}
\Pr(X\ge 10\ln n)&\le\Pr\left(X\ge (1+4)\mathbb{E}X\right)\\
&\le\exp\left(-\frac{4^2}{2+4}\mathbb{E}X\right) 
=\exp\left(-\frac{16}{6}2\ln n\right)=n^{-16/3}.
\end{align*}
By the union bound the probability that in 
a given subseries event $A$ occurs for any agent more
than $10\ln n-1$ times does not exceed $n^{-13/3}$.
Thus, for $n$ large enough and $k>\frac{|I|}{10n\log n}>\frac{6}{13}\eta$
the probability that at least half of $k$ subseries correspond to parallel time smaller than $\ln n$
is less than
$$\binom{k}{k/2} \left(n^{-13/3}\right)^{k/2}< 2^kn^{-13k/6}<n^{-\eta}.$$
As this event occurs only with negligible probability, we have parallel time  at least $\frac{2|I|}{n}$ whp.
\end{proof}

\subsection{Leader Election}\label{LE}

In {\em leader election} 
(with confirmation) 
in the initial configuration at least one agent is a candidate to become a {\em unique leader}, and all other agents begin as followers. 
The main goal in leader election is to distinguish and report selection of the unique leader, and to declare all other agents as followers. 
The state space of the leader election protocol presented below is $S=\{L,L^*,F,F^*\},$
where all initial leader and follower candidates are in states $L$ and $F,$ respectively.
The remaining states include $L^*$ referring to the confirmed unique leader,
and
$F^*$ utilised by confirmed followers.
The state space is partitioned into two groups
$\mathcal{G}_0=\{L,L^*,F^*\}$ and $\mathcal{G}_1=\{F\}.$

{\bf LE-protocol:} As at least one agent starts in state $L,$
these agents target group $\mathcal{G}_0$ using a double rule (1) and (2) and when state $L^*$ is eventually reached, with exactly one agent being in this state, the epidemic process defined by the transition rules (3)-(4) informs all followers about  successful leader election.\nopagebreak
\begin{small}
\begin{tcolorbox}
\begin{multicols}{2}
\noindent
\begin{alignat*}{5}
\transition{1}{L}{\mathcal{G}_0|L}{L}{F}
\transitionnull{2}{L}{\mathcal{G}_0|null}{L^*}{}
\end{alignat*}
\begin{alignat*}{5}
\transition{3}{L^*}{\mathcal{G}_1|F}{L^*}{F^*}
\transition{4}{F^*}{\mathcal{G}_1|F}{F^*}{F^*}
\end{alignat*}
\end{multicols}
\vspace{-3em}
\end{tcolorbox}
\end{small}

\begin{lemma}\label{LE-lemma}
The fragmented time of LE-protocol is
$O(\log n)$ whp.
\end{lemma}
\begin{proof}
As described above, the execution of the leader election protocol can be divided into two consecutive phases. In the first phase, driven by rules (1) and (2), a unique agent in the leader state $L^*$ is computed. In the second phase, governed by rules (3) and (4), the states of all followers are changed to $F^*$. We prove that the fragmented time of each phase is $O(\log n)$.

In the first phase, the number of interactions $T$, in which a single leader is computed, can be expressed as $T = T_n + T_{n-1} + \cdots + T_1$.

In this formula, $T_i$ for $i>1$ is the number of expected interactions to eliminate a leader candidate when $i$ candidates are still in contention, and $T_1$ is the number of interactions in which a single $L$ transitions to $L^*$. Note that $T_i$ has geometric distributions, and $\mathbb{E}[T_i] = n/i$. So, by Janson's bound~\cite{geometric}, this number of interactions is $O(n\log n)$ with high probability. Since, in the first phase, all interactions are internal, the fragmented parallel time is $O(\log n)$.

In the second phase, the number of interactions $T$ can be expressed as the sum $T = T_1 + T_2 + \cdots + T_{n-1}$.

In this formula, $T_i$ stands for the number of interactions required to enlarge $\mathcal{G}_0$ with $i$ agents by one agent. Note that $T_i$ has geometric distributions, and $\mathbb{E}[T_i] = n/i$. So, by Janson's bound, this number of interactions is $O(n\log n)$ with high probability. In the second phase, each external interaction targets $\mathcal{G}_1$ and reduces this group by 1. Therefore, by Lemma \ref{YouOnlyTakeOnceLemma}, the fragmented time is also $O(\log n)$.

\end{proof}

\subsection{Majority Computation}\label{Majority}

The state space of the majority protocol is $S=\{G,G^*,R,R^*,N\}$,
and
in the initial
configuration each agent is either in state $G$ or $R$.  
The main goal is
to decide which subpopulation of agents in state $G$ or $R$ is greater than the other. If the subpopulation in state $G$ is greater, all agents are expected to stabilise in state $G^*.$ Otherwise, they must stabilise in state $R^*.$ The majority protocol {\bf M-protocol} described below uses also neutral state $N$.
The state space is partitioned into three groups
$\mathcal{G}_R=\{R,R^*\}$, $\mathcal{G}_N=\{N\},$ and $\mathcal{G}_G=\{G,G^*\}.$

{\bf M-protocol} has the following transition rules: \nopagebreak
\begin{small}
\begin{tcolorbox}
\begin{multicols}{2}
\noindent
\begin{alignat*}{5}
\transition{1}{R}{\mathcal{G}_G|G}{N}{N}
\transitionnull{2}{R}{\mathcal{G}_G|null}{R^*}{}
\transition{3}{G}{\mathcal{G_R}|R}{N}{N}
\end{alignat*}
\begin{alignat*}{5}
\transitionnull{4}{G}{\mathcal{G_R}|null}{G^*}{}
\transition{5}{R^*}{\mathcal{G_N}|N}{R^*}{R^*}
\transition{6}{G^*}{\mathcal{G_N}|N}{G^*}{G^*}
\end{alignat*}
\end{multicols}
\vspace{-3em}
\end{tcolorbox}
\end{small}

Transition rules (1) and (3) instruct agents in states $R$ and $G$ to become neutral for as long as pairs $R+G$ and $G+R$ can be formed. As soon as one of these states is no longer present in the population either rule (2) or (4) is used to change state $R$ to $R^*$ or $G$ to $G^*,$ respectively. In addition, either rule (5) or (6) is used to change neutral state to $R^*$ or $G^*,$ respectively.
Alternatively, if all states $G$ and $R$ disappear after application of rules (1) and (3), the population stabilises in the neutral state 
$N.$
\begin{lemma}\label{MP}
The fragmented time of M-protocol is
$O(\log n)$ whp.
\end{lemma}

\begin{proof}
The stabilisation time of the majority protocol is divided into two phases.
In the first phase at least one of the states $G$ or $R$
gets expunged.
In case of imbalance between $G$ and $R$ states we also have the second phase in which all agents either in states $N$ and $G$ or $N$ and $R$ change into $G^*$ or $R^*,$ respectively.

In the first phase the number of interactions is $T=T_m+T_{m-2}+T_{m-4}+\cdots +T_{k}$, 
where $m$ is the initial number $G$ and $R$ agents and $k$ is the bias between the colors.
Each random variable $T_i$ is the time needed to neutralise the next pair $G,R$, when still $i$ pairs of agents $G,R$ can be formed.
Note that $T_i$ has geometric distribution and $\mathbb{E}T_i=n/i$.
So by Janson bound $T$ is $O(n\log n)$ whp.
Since in any external interaction the responder gets
to $\mathcal{G}_N$, any agent can be the responder at most once in an external interaction.
By Lemma \ref{YouOnlyTakeOnceLemma} the parallel time of the first phase is $O(\log n)$.

In the second phase the number of interactions
needed to make all transitions $G\rightarrow G^*$ or $R\rightarrow R^*$ is $T=T_k+T_{k-1}+T_{k-2}+\cdots +T_1$,
where $T_i$ is the number of interactions to reduce the number of
$G$ or $R$ agents from $i$ to $i-1$.
Note that $T_i$ has geometric distribution and $\mathbb{E}T_i=n/i$.
Thus, by Janson bound $T$ is $O(n\log n)$ whp.
Then we count the number of interactions $T'$ needed to
turn all remaining agents to state $G^*$ or $R^*$
after initial $T$ interactions of the second phase.
We get
$$T'=T_c'+T_{c+1}'+T_{c+2}'+\cdots +T_{n-1}',
$$
where $c$ is the number of $G^*$ or $R^*$ agents after $T$ interactions of the second phase.
Note that $T_i'$ has geometric distribution and $\mathbb{E}T_i'=n/i$.
So by Janson bound $T$ is $O(n\log n)$ whp.
Thus the second phase has $O(n\log n)$ interactions whp.
Any agent can be a responder at most once in an external interaction of phase 2.
By Lemma \ref{YouOnlyTakeOnceLemma} the parallel time of the second phase is $O(\log n)$.    
\end{proof}

\section{Computing the Median}\label{s3}

In this section, we consider computing the median of $n$ distinct keys, each of which is held by one of the $n$ agents.
For agents $a,b$ belonging to the set $S$ of agents, the relation $b<a$ denotes $\mbox{key}(a)<\mbox{key}(b)$.
We adopt here a \emph{comparison model} in which the transition function depends not only on the states of the agents, but also on the order of their keys.
The keys are hidden and there is no other way to access them.
The number of states remains fixed.
A similar limited use of large keys can be found in {\em community protocols} in \cite{DBLP:conf/icalp/GuerraouiR09} to handle Byzantine failures.

For any agent $c \in S,$ let $\mathbb{A}_c$ and $\mathbb{B}_c$ be the set of 
all agents above and below $c$ respectively.
The agent $m$ is the unique {\em median} if $|\mathbb{B}_m|-|\mathbb{A}_m|=0,$ for odd $n,$ 
or one of the two medians if $||\mathbb{B}_m|-|\mathbb{A}_m||=1,$ for even $n$.
In this version we assume that all keys are different and $n$ is odd.
The arbitrary case
requires minor amendments, as the answer may refer to two agents. 
Before we consider selective protocols, we first consider median computation
in comparison model with a standard random scheduler.

\subsection{Median Computation in Standard Model}

\begin{theorem}\label{LBmedian}
Finding the median in the comparison model requires $\Omega(n)$ time in expectation.
\end{theorem}

\begin{proof}
Assume that $n$ is odd and agents $a_1$ and $a_2$ share between themselves the median and the key immediately succeeding it (in the total order of keys).
One can observe that before the first interaction between $a_1$ and $a_2$,
all consecutive configurations of states of all agents are independent 
from whether $a_1$ or $a_2$ is associated with the median.
Thus before the first interaction between $a_1$ and $a_2$ 
no algorithm can declare either of them as the median.
And since the expected number of interactions preceding 
the first interaction between $a_1$ and $a_2$ is $\Omega(n^2),$
the thesis of the theorem follows.
\end{proof}

Now we formulate an almost optimal median population protocol in the adopted model.
All~agents start this protocol in neutral state $N.$
In due course, agents change their states,
s.t., eventually all agents associated with keys smaller than the median 
end up in state $B,$ those with keys greater than the median in state $A$,
and the median conclude in state $N.$ 
The median protocol uses the following symmetric transition function:
\setlength\columnsep{24pt}
\begin{small}
\begin{tcolorbox}[left=0pt]
\begin{multicols}{2}
\noindent
\begin{alignat*}{5}
\transitionrelComment{1}{N}{N}{B}{A}{\text{initialisation}}
\transitionrelComment{2}{A}{B}{B}{A}{\text{fix order}}
\end{alignat*}
\begin{alignat*}{5}
\transitionrelComment{3}{A}{N}{N}{A}{\text{fix order}}
\transitionrelComment{4}{N}{B}{B}{N}{\text{fix order}}
\end{alignat*}
\end{multicols}
\vspace{-3em}
\end{tcolorbox}
\end{small}
Note that there is always the same number of agents in states $B$ and $A$
and one agent will remain in state $N$ as $n$ is an odd number.


\begin{theorem}
    The median protocol operates in $O(n\log n)$ time both in expectation and whp.
\end{theorem}

\begin{proof}
    We say that a pair of agents $a$ and $b$ is {\em disordered} if an interaction 
    between $a$ and $b$ is meaningful.
    We define the {\em disorder} $d(C)$ of a configuration $C$ 
    as the total number of disordered pairs of agents in this configuration.
    Since all agents start in state $N,$ any initial interaction is meaningful, via application of rule (1). 
    And in turn the disorder of the initial configuration is $d(C_0)=\binom{n}{2}$.
    In the final configuration $C_{\infty},$ when the agent with the median key is in state $N,$
    and all agents with smaller and larger (than the median) keys are in states $B$ and $A$ respectively, 
    the disorder $d(C_{\infty})=0$, as none of the rules can be applied.
    
    \begin{proposition}\label{dec_d}
    Any meaningful interaction reduces the disorder of a configuration. 
    \end{proposition}
    
\begin{proof}
Let us order the agents according to the values of their keys: $a_1,a_2,\ldots,a_n$.
    Consider a meaningful interaction between $a_i$ and $a_j,$ 
    where $i<j,$
    which changes the configuration from $C$ to $C'$.
    We observe that the pair $a_i,a_j$ is no longer disordered. 
    We also show that for any other agent $a_k,$ the number of 
    disordered pairs of the form $a_k,a_i$ and $a_k,a_j$ does not increase.
    There are two types of meaningful interactions between $a_i$ and $a_j$.
    The first type refers to rule $(1)$ and the second swaps the states
    of $a_i$ and $a_j$ by either of the rules $(2),(3)$ or $(4)$.
    
    Assume first that $k<i$. 
    If the interaction between $a_i$ and $a_j$ is of the second type, the
    level of disorder in pairs $a_k,a_i$ and $a_k,a_j$ remains the same.
    Consider now the interaction between of the first type.
    If the state $s(a_k)$ of agent $a_k$ is $B$, then both pairs $a_k,a_i$ and $a_k,a_j$ are not disordered in $C'$.
    If $s(a_k)=A$ or $N$, then both pairs $a_k,a_i$ and $a_k,a_j$ are disordered in $C$.
    Thus in both sub-cases the disorder in pairs $a_k,a_i$ and $a_k,a_j$ 
    cannot increase.
    In the second case $k>j,$ and the reasoning is analogous, as this case is symmetric to the first one.
    
    The remaining case is $i<k<j$.
    If the interaction between $a_i$ and $a_j$ is of the first type
    or swaps states $B$ and $A$, then there is no disordered pair of the form $a_k,a_i$ or $a_k,a_j$ in $C'$.
    So it remains to consider interactions between $a_i$ and $a_j$
    that swap the states $N$ with $B$ or $A$.  
    Without loss of generality, assume that the swap is between $N$ and $B$.
    If $s(a_k)=B,N$, then neither of the pairs $a_k,a_i$ or $a_k,a_j$ is disordered in $C'$.
    Finally, if $s(a_k)=A$, then only pair $a_k,a_j$ is disordered both is $C$ and $C'$.
\end{proof}

The probability of making a meaningful interaction in a configuration $C,$ s.t., $d(C)=i$ is $p_i=i/\binom{n}{2}$. Let the random variable $T_i$ be the number of interactions needed to observe a meaningful interaction for a configuration $C$ when $d(C)=i$. We have $\mathbb{E}[T_i]={\binom{n}{2}}/i$. The expected number of interactions $\mathbb{E}[T]$ to transition from $C_0$ to $C_{\infty}$ is
\[ \mathbb{E}[T] \le \mathbb{E}[T_1] + \mathbb{E}[T_2] + \cdots + \mathbb{E}\left[T_{\binom{n}{2}}\right] = \binom{n}{2}H_{\binom{n}{2}} = O(n^2\log n). \]

One can also prove that $T=O(n^2\log n)$ whp applying Janson's bound. However, we show here an alternative proof by a potential function argument.

In particular, we show that for two subsequent configurations $C$ and $C'$ we have
\[ \mathbb{E}[d(C')] \le \left(1-\frac{1}{\binom{n}{2}}\right)d(C). \]
This inequality becomes trivial when $d(C)=0.$ Otherwise, when $d(C)\neq 0$,
\[ \mathbb{E}[d(C')] \le \left(1-\frac{d(C)}{\binom{n}{2}}\right)d(C) + \frac{d(C)}{\binom{n}{2}}\left(d(C)-1\right) = \left(1-\frac{1}{\binom{n}{2}}\right)d(C). \]
After $t$ interactions beyond configuration $C_0$ we get
\[ \mathbb{E}[d(C_t)] \le \left(1-\frac{1}{\binom{n}{2}}\right)^t d(C_0) \le \exp\left(-\frac{t}{\binom{n}{2}}\right)\binom{n}{2}. \]
We obtain $\mathbb{E}[d(C_t)]<n^{-\eta}$ for
\[ t \ge \binom{n}{2}\left(\ln \binom{n}{2}+\ln\eta\right). \]
And finally, when $\mathbb{E}[d(C_t)]<n^{-\eta}$, by Markov's inequality, we get $\Pr(d(C_t)\ge 1)<n^{-\eta}$. This is equivalent to $\Pr(C_t\neq C_{\infty})<n^{-\eta}$.

\end{proof}
A similar reasoning can be found in the context of probabilistic sorting, see~\cite{Murty, Young}.

\subsection{Fast Median Computation}
 
We present and analyse here efficient median computation in selective population protocols.
The proposed solution is done by breaking up the full protocol into smaller blocks implemented as independent stabilisation processes with clearly defined inputs 
and outputs, 
as well as efficient and stable solutions. Each of these independent processes (see Fast-median algorithm below) including leader (pivot) election, partitioning agents wrt the key of the pivot, and majority computation, is executed on distinct partitions of states.
Recall from Section~\ref{SM} that the leader elected in the beginning of the computation process executes the code of the solution embedded in the transition function, and manages all input/output operations.\\

\begin{algorithm}[ht!]\small 
\begin{algorithmic}[1]
\Require{$S$ -- all agents set, $C=S$ -- median candidate set}
\State Select randomly leader (as pivot) $p\in S$ \Comment{leader election}
\Repeat
\State Partition $S$ to $\mathbb{B}_{p}=\{x\in S:x<p\},\mathbb{A}_{p}=\{x\in S:x>p\},\{p\};$  

\State \textbf{switch} \Comment{majority computation}
\Nest
\State \textbf{case} $(|\mathbb{B}_p|>|\mathbb{A}_p|)\longrightarrow C=C\cap \mathbb{B}_p$
\State \textbf{case} $(|\mathbb{B}_p|<|\mathbb{A}_p|)\longrightarrow C=C\cap \mathbb{A}_p$
\State \textbf{case} $(|\mathbb{B}_p|=|\mathbb{A}_p|)\longrightarrow C=\{p\}$
\EndNest
\State $p\leftarrow$ randomly chosen (by current $p$) agent in $C$ \Comment{leader hand over}
\Until{ $(|\mathbb{B}_p|=|\mathbb{A}_p|)$}
\State \Return $p$ \Comment{result announcement}
\end{algorithmic}
\caption{Fast-median.}\label{alg:semisplit-mcts}
\end{algorithm}

Recall that leader election and majority computation, were discussed in Sections~\ref{LE} and \ref{Majority}, respectively. 
Thus the focus in this section is on efficient partitioning of all agents in $S$ to $\mathbb{B}_{p}=\{x\in S:x<p\},\mathbb{A}_{p}=\{x\in S:x>p\}$, and $\{p\}.$
Note that such partitioning is not trivial as due to the restrictions in the model the pivot $p$ cannot distribute the value of its key to all agents in the population. Instead, the agents gradually learn their relationship with respect to the pivot by comparing their keys with other agents.

\begin{theorem}\label{fast-median}
    Algorithms Fast-median operates in fragmented parallel time $O(\log^4 n)$ whp.  
\end{theorem}

\begin{proof}
    The proof can be found at the end of Section~\ref{median:coloring}.
\end{proof}
\subsubsection{Partitioning via Coloring}\label{median:coloring}
We commence by providing an overview of the argument.
The partitioning of all agents occurs in multiple phases, represented by consecutive stabilization processes. The objective in each phase is to correctly partition a constant fraction of agents that have not been partitioned yet. Each phase has a time complexity of $O(\log^2 n)$. Since partitioning requires $O(\log n)$ phases, the overall time complexity becomes $O(\log^3 n)$.

In the median protocol, we execute $O(\log n)$ partitioning steps, which constitute the majority of the computation time. Consequently, the total time complexity for computing the median is $O(\log^4 n)$.
To analyze a single phase of partitioning, we categorize the set of uncolored agents above the pivot into $2\log n$ buckets, and similarly for those below the pivot. We then demonstrate that within $O(\log^2 n)$ time, the algorithm successfully colors $\log n$ agents from the first bucket.
In each successive time period indexed by $i = 2, 3, \ldots$, with a duration of $O(\log n)$, the algorithm colors $2^{i-1}\log n$ agents in the $i$-th bucket.
Consequently, after $O(\log^2 n)$ time from the initiation of a phase, a constant fraction of uncolored agents acquires colors.

The input for the partitioning process consists of the leader agent in the pivot state $P$ and all other agents in the state $N_{in}$.
We utilize two groups of states:  $\mathcal{G}=\{P,B_0,\ldots,B_{21},A_0,\ldots,A_{21},N\}$ and $\mathcal{G}_{in}=\{N_{in}\}$. We interpret states $A_t$, $B_t$, and $N$ as above, below, and neutral colors, respectively. An agent adopts state $A_t$ ($B_t$) as soon as it learns that its key is above (below) the key of the pivot.

During each phase, we color approximately a fraction of $1/22$ of yet uncolored agents. Upon the conclusion of the phase, these agents are moved to a different group, i.e., they do not participate in the partitioning of uncolored agents in the remaining phases.

In order to limit the activity of colored agents and, in turn, the duration of each phase, we introduce the concept of \emph{tickets}.
While the pivot has an unlimited number of tickets, any newly colored agent receives a fixed pool of 21 tickets.
For as long as any colored agent has tickets, it targets agents in group $\mathcal{G}_{in}$ trying to color them.
During such an interaction, a colored agent loses one ticket and moves one agent from $\mathcal{G}_{in}$ to $\mathcal{G}$.
Once a colored agent loses all its tickets, it starts targeting group $\mathcal{G}$.
The (partial) coloring phase concludes when the group $\mathcal{G}_{in}$ becomes empty.
The set of relevant rules is given below.
\setlength\columnsep{14pt}

\begin{small}
\begin{tcolorbox}[left=-2pt]
\begin{multicols}{2}
\noindent
\begin{alignat*}{5}\label{Coloring_Phase_Protocol}
\transitionrel{1}{P}{\mathcal{G}_{in}|N_{in}}{P}{A_{21}}
\transitionrel[>]{1}{P}{\mathcal{G}_{in}|N_{in}}{P}{B_{21}}
\transitionrel[<]{2}{B_{t>0}}{\mathcal{G}_{in}|N_{in}}{B_{t-1}}{N}
\transitionrel[>]{2}{B_{t>0}}{\mathcal{G}_{in}|N_{in}}{B_{t-1}}{B_{21}}
\transitionrel[>]{2}{B_{0}}{\mathcal{G}|N}{B_0}{B_{21}}
\transitionrel[>]{3}{A_{t > 0}}{\mathcal{G}_{in}|N_{in}}{A_{t - 1}}{N}
\end{alignat*}
\begin{alignat*}{5}
\transitionrel[<]{3}{A_{t >0}}{\mathcal{G}_{in}|N_{in}}{A_{t-1}}{A_{21}}
\transitionrel[<]{3}{A_{0}}{\mathcal{G}|N}{A_0}{A_{21}}
\transitionrel[<]{4}{N}{\mathcal{G}|P}{B_{21}}{P}
\transitionrel[>]{4}{N}{\mathcal{G}|P}{A_{21}}{P}
\transitionrel[<]{4}{N}{\mathcal{G}|B_t}{B_{21}}{B_t}
\transitionrel[>]{4}{N}{\mathcal{G}|A_t}{A_{21}}{A_t}
\end{alignat*}
\end{multicols}
\vspace{-3em}
\end{tcolorbox}
\end{small}

We formulate here a tail bound that works for hypergeometric sequences for the case of small fraction $p$ of black balls.
We need this more sensitive bound in some of our proofs.

    \begin{lemma}\label{hyper}
    Assume we have an urn with $n$ balls where $pn$ of them are black. Let $X_i$ be a binary random variable equal to one iff in the $i$th draw without replacement the drawn ball was black, and let $X = \sum_{i=1}^{\kappa} X_i$. Then, for ${\kappa} \leq n$ and $0<\delta < 1,$ we get
    $$ \Pr\left(  X < (1-\delta)p{\kappa}-p  \right) < {\kappa}\exp\left(\frac{-\delta ^2p{\kappa}}{2}\right). $$
    \end{lemma}

    \begin{proof}
    If, after drawing $i-1$ balls, no more than $(i-1)p$ of them are black, we have $q_i = \Pr(X_i=1|\text{current urn content}) \ge p$.

    Now, let's define the random variables $Y_i$ as follows:
    \begin{itemize}
        \item If, after drawing $i-1$ balls, more than $(i-1)p$ of them are black, we independently draw $Y_i=1$ with a probability of $p$, and $Y_i=0$ otherwise.
    
        \item If, after drawing $i-1$ balls, at most $(i-1)p$ of them are black and $X_i=0$, then $Y_i=0$.
    
        \item If, after drawing $i-1$ balls, at most $(i-1)p$ of them are black and $X_i=1$, then $Y_i=1$ with a probability of $p/q_i$, and $Y_i=0$ otherwise.
    \end{itemize}
    
        Variables $Y_1,Y_2,\ldots,Y_{\kappa}$ are independent,
        because no $Y_i$ depends on any combination of $Y_j:j<i$, and $\Pr(Y_i)=p$.
        If $X < (1-\delta)p{\kappa}$, then let $j\in[0,{\kappa})$ be the last draw before which not less than $(j-1)p$ balls drawn are black.
        This implies $Y_i\le X_i$ for $i>j$ and $X_1+\cdots+X_j\ge p(j-1)$.
        Thus when $X_1+\cdots +X_{\kappa}<(1-\delta) p{\kappa}-p$, we get 
        \begin{align*}
          X_{j+1}+\cdots+X_n & = 
          (X_1+\cdots +X_{\kappa})-(X_1+\cdots +X_j)\\
          & <(1-\delta)p{\kappa}-pj=\left(1-\delta\frac{{\kappa}}{{\kappa}-j}\right)p({\kappa}-j). 
        \end{align*}
        Also $Y_{j+1}+\cdots+Y_{\kappa}\le X_{j+1}+\cdots+X_{\kappa}<(1-\delta) p{\kappa}-pj$.
        In other words, we can say that the existence of such $j$ that $Y_{j+1}+\cdots+Y_{\kappa}<(1-\delta) p{\kappa}-pj$ is
        a necessary condition for $X < (1-\delta)p{\kappa}-p$.
        Now we upper bound the probability of existence of such $j$ where $Y_{j+1}+\cdots+Y_{\kappa}<(1-\delta) p{\kappa}-pj$.
        For a fixed $j$ by Chernoff inequality we get
        \begin{align*}
        &\Pr\left(Y_{j+1}+\cdots +Y_n<\left(1-\delta\frac{{\kappa}}{{\kappa}-j}\right)p({\kappa}-j)\right)\\
        &
            <\exp\left(-\frac{{\kappa}^2\delta^2p({\kappa}-j)}{2({\kappa}-j)^2}\right)\le
            \exp\left(\frac{-\delta ^2p{\kappa}}{2}\right).
        \end{align*}
        By the union bound on $j=0,1,\ldots,{\kappa}-1$ we get the thesis of the lemma.
    \end{proof}

    Denote the number of agents participating (not previously colored) in the phase by~$m$.
    For any interaction $t,$ let $k(t)$ be a number of agents in group $\mathcal{G}$ in $t$
    and $Inf_x(t)$ be a number of informed agents which are in  states $A$ or $B$ in bucket $x$.
    The sequence of useful technical lemmas leading to the thesis of  Theorem~\ref{coloring_theorem} follows.

    \begin{lemma}\label{rlog2n} \label{claim}
    If during interaction $t$ the number of colored agents is $r\ge\log^2 n$, then $k(t+500n)\ge\min\{20r,m\}$ whp.
    \end{lemma}
    
    \begin{proof}
            Any colored agent has $21$ tickets to utilise.
        As long as in interaction $i$ there are at least $r/21$ agents possessing tickets, the probability $q_i$ of choosing one of them to be the initiator and distributing a ticket is
        at least $\frac{r}{21n}$.
        Note that when less than $r/21$ agents have tickets, then already at least $21\cdot 20r/21=20r$ tickets are utilised.
        We define a random variable $X_i\in\{0,1\}$, s.t., $X_i=1$ iff a ticket is utilised in interaction $i$.
        Let $Y_i$ be a random variable defined as follows:
        \begin{itemize}
            \item if $q_i\ge\frac{r}{21n}$ and $X_i=0$, then $Y_i=0$,
            \item if $q_i\ge\frac{r}{21n}$ and $X_i=1$, then $Y_i=1$ with probability $\frac{r}{21 q_i n},$ and $Y_i=0$ otherwise,
            \item if $q_i<\frac{r}{21n}$, then $Y_i=1$ with probability $\frac{r}{21n},$   and $Y_i=0$ otherwise.
        \end{itemize}
        Random variables $Y_i$ are independent and $Pr(Y_i=1)=\frac{r}{21n}$.
        Now we show that if
        $$Y=Y_{t+1}+\cdots +Y_{t+500n}\ge 20r=\frac{21}{25}\mathbb{E}Y, \; \mbox{ where } \mathbb{E}Y=\frac{500r}{21}$$
        then $k(t+500n)\ge\min\{20r,m\}$.
        The first possibility is that in some interaction $i,$ we have $q_i<\frac{r}{21n}$.
        This implies that in interaction $i$ already $20r$ tickets are utilised.
        If~on the other hand, in all interactions $i$ we have $q_i\ge\frac{r}{21n}$,
        we guarantee choosing $20r$ times an agent possessing a ticket to be an initiator of an interaction.
        Thus in both cases $k(t+500n)\ge\min\{20r,m\}$.
        
        By Chernoff inequality
        \[
        \Pr\left(Y < \frac{21}{25} \mathbb{E}Y\right) < \exp\left(-\frac{4^2}{25^2} \frac{500}{21}r/2\right) < \exp\left(-\frac{32}{105} \log^2 n\right) < n^{-\eta}.
        \]
        So $k(t+500n)\ge\min\{20r,n\}$ whp.
    \end{proof}

\begin{lemma}\label{onephase}
    The fragmented time of one phase of partitioning by coloring is $O(\log^2 n)$ whp.     
\end{lemma}

\begin{proof}
    Recall that the number of agents participating in the phase is denoted by~$m$.
    Without loss of generality, assume that there are more agents with keys greater than the pivot's key.
    Let us partition the agents with keys larger than the pivot's key into $2\log n$ disjoint buckets of equal size numbered $i=1,2,\ldots,2\log n$, s.t., any bucket $i$ contains agents with keys greater than those in buckets with numbers smaller than $i$.
    
    We divide the sequence of interactions into interaction periods $[t_0,t_1),[t_1,t_2),[t_2,t_3), \dots$
    such that the first interaction period has length $O(n\log^2 n),$ and each subsequent period is of length $O(n\log n)$.
    We show that in interaction $t_{\log n}$ all $m$ agents are in the group $\mathcal{G}$ whp.
    Each colored agent can utilise at most 21 tickets to move other agents to $\mathcal{G}$.
    This implies that in $t_{\log n}$ at least $m/22$ of these agents have color.
    Since each interaction between groups removes an agent from ${\mathcal G}_{in}$ irreversibly,
    and there are $O(n \log^2 n)$ interactions altogether whp,
    by Lemma \ref{YouOnlyTakeOnceLemma} the fragmented time of one phase is $O(\log^2 n)$.
    The thesis of the lemma is a direct consequence of Lemma \ref{doubleInductiveColoringLemma}, which we formulate below.
    And indeed, since this implies $k(t_{\log n}) \ge\min\left\{20\cdot 2^{\log n-1}\log^2 n, m\right\},$ we get $k(t_{\log n})=m$.
    \end{proof}
    
    \begin{lemma}\label{last}\label{doubleInductiveColoringLemma}
        There is a constant $c>0,$ s.t., if $|[t_0,t_1)|=cn\log^2 n$ and $|[t_{i-1},t_i)|=cn\log n$ for all $i>1$, then whp
        \begin{itemize}
            \item $k(t_i) \ge\min\left\{20\cdot 2^{i-1}\log^2 n, m\right\}$
            \item $Inf_i(t_i)>\min\left\{2^{i-1}\log n, \frac{m}{10\log n}\right\}$
        \end{itemize} 
    \end{lemma}
    
    \begin{proof}
        We proceed by induction on $i$.
        
        We first deal with the base case $i=1$.
        Let $t=t_0+cn\log^2 n$, where a constant $c>0$ is to be specified later on, and $t_1=t+500n$.
        Note that the pivot always colors an uncolored agent during an interaction.
        During interaction period $[t_0,t)$ the pivot is chosen as the initiator at least $\log^2 n$ times whp.
        Indeed, by Chernoff bound the number of times $X,$ for which the pivot is chosen to be an initiator, satisfies for some $c>0$
        $$\Pr(X<\log^2 n)<\Pr(X<(1/c)\mathbb{E}X)\le e^{-(c-1)^2\log^2 n/2c}<n^{-\eta}.$$
        
        If in $[t_0,t)$ group ${\mathcal G}_{in}$ runs out of agents, then $k(t)=m$.
        Otherwise in time $t$ at least $\log^2 n$ agents are colored and by lemma \ref{claim} we have $k(t_1)=k(t+500n)\ge\min\left\{20\log^2 n, m\right\}$ whp.
        
        On the other hand, consider  an uncolored agent from the first bucket.
        No matter whether it belongs to ${\mathcal G}_{in}$ or ${\mathcal G}$, it has probability at least $\frac{1}{nm}$ of interaction with the pivot.
        Due to the way the buckets are defined, see the proof of Lemma~\ref{onephase}, each of them has at least $\frac{m}{4\log n}$ agents.
        If in any interaction of $[t_0,t)$ more than $\frac{m}{10\log n}$ agents are colored, then $Inf_i(t_1)>\min\left\{\log n, \frac{m}{10\log n}\right\}$.
        Otherwise, the probability  of coloring a new agent from the first bucket in each interaction of $[t_0,t)$ is at least $\frac{3}{20n\log n}$.
        By~Chernoff bound, for some $c>0$, the number $X$ of colored agents in the first bucket in $t$ satisfies 
        $$\Pr(X\le\log n)<\Pr(X\le (2/3c) \mathbb{E}X)\le e^{-(3c-2)^2\log^2 n/18c}<n^{-\eta}.$$         
        This proves that $Inf_i(t_1)>\min\left\{\log n, \frac{m}{10\log n}\right\}$ whp.
        
        Now we show the inductive step.
        By the inductive hypothesis we have whp
        $$k(t_{i-1}) >\min\left\{20\cdot 2^{i-2}\log^2 n, m\right\}$$ and
        $$Inf_{i-1}(t_{i-1})>\min\left\{2^{i-2}\log n, \frac{m}{10\log n}\right\}.$$
        
        By Lemma \ref{hyper}, for any $\tau\ge t_{i-1}$ there are whp at least $\frac{0.9k(\tau)}{4\log n}$ agents from $i$th bucket in~${\mathcal G}_1$.
        If~at most half of these agents are colored, the probability of coloring a new agent from the first bucket in $\tau$ is at least
        $$  q= \frac{0.9k(\tau)}{8n\log n} \cdot \frac{Inf_{i-1}(t_{i-1})}{k(\tau)}\ge
               \frac{Inf_{i-1}(t_{i-1})}{10n\log n}\ge
               \min\left\{\frac{2^{i-2}}{10n}, \frac{1}{10\log^2 n}\right\}.$$
        
        The necessary condition for half of these agents to be colored in $\tau$ is
        $$Inf_i(\tau)>\frac{0.9k(\tau)}{8\log n}\ge\frac{k(t_{i-1})}{10\log n}\ge
          \min\left\{2^{i-1}\log n, \frac{m}{10\log n}\right\}.$$
          
        Let $t=t_{i-1}+cn\log n$, where a constant $c$ will be specified later, and $t_i=t+500n$.
        Let $q_{\tau}$ be the probability of coloring a new agent from the first bucket in $\tau$ given the configuration in $\tau$.
        We attribute to each interaction $\tau\in[t_{i-1},t)$ a 0-1 random variable $Y_{\tau}$ as follows.
        \begin{itemize}
            \item If in $\tau$ the number of uncolored agents in the $i$th bucket is at least $\frac{0.9k(\tau)}{8n\log n}$
                and no new agent from $i$th bucket is colored, then $Y_{\tau}=0$
            \item If in $\tau$ the number of uncolored agents in the $i$th bucket is at least $\frac{0.9k(\tau)}{8n\log n}$
                and a new agent from $i$th bucket is colored, then $Y_{\tau}=1$ with probability $q/q_{\tau}$ and $Y_{\tau}=0$ otherwise.
            \item If in $\tau$ the number of uncolored agents in the $i$th bucket is smaller than $\frac{0.9k(\tau)}{8n\log n}$,
                  then we choose $Y{\tau}=1$ with probability $q$.
        \end{itemize}
        For each $\tau\in[t_{i-1},t)$ we have $\Pr(Y_{\tau}=1)=q$ and variables $Y_{\tau}$ are independent. 
        Note that if $Y=\sum_{\tau\in[t_{i-1},t)} Y_{\tau}$ satisfies 
        $$Y\ge y = 20qn\log n\ge\min\left\{2^{i-1}\log n, \frac{n}{10\log n}\right\},$$
        then at least $y$ agents from the $i$th bucket are colored in $[t_{i-1},t)$ or in some $\tau$
        the number of uncolored agents in $\tau$ is smaller than $\frac{0.9k(\tau)}{8n\log n}$.
        Both cases imply
        $$Inf_{i}(t)>\min\left\{2^{i-1}\log n, \frac{n}{10\log n}\right\}.$$
        Since $\mathbb{E}Y=cqn\log n$, by Chernoff inequality there is $c>0,$ such that $\Pr(Y<y)$ is not greater than
        $$\Pr\left(Y<\left(1-\frac{c-20}c\right)cqn\log n\right)
        \le e^{-(c-20)^2qn\log n/2c}\le n^{-\eta}.$$
        Thus we proved that whp
        $$Inf_{i}(t_{i-1})>\min\left\{2^{i-1}\log n, \frac{m}{10\log n}\right\}.$$
        The same bound holds for any of $\log n$ last buckets.
        This guarantees that the number of colored agents in interaction $t$ is at least
        $$\min\left\{2^{i-1}\log n, \frac{m}{10\log n}\right\}\cdot\log n=\min\left\{2^{i-1}\log^2 n, \frac{m}{10}\right\}.$$
        By Lemma \ref{claim} we have whp
        \begin{align*}
            k(t+500n)&\ge\min\left\{20\min\left\{2^{i-1}\log^2 n, \frac{m}{10}\right\} ,m\right\}\\
            &\ge \min\left\{20\cdot 2^{i-1}\log^2 n, m\right\}
        \end{align*}

    \end{proof}

The phase ends when group $\mathcal{G}_{in}$ becomes empty. 
Each agent is relocated from $\mathcal{G}_{in}$ to $\mathcal{G}$ only if it either gets properly colored or it was given a ticket. Since each colored agent has only 21 tickets to utilise, we can formulate the following fact.

\begin{fact}
    After single coloring phase a fraction of at least $\frac{1}{22}$ uncolored agents gets colored. 
\end{fact}

Thus, after $O(\log n)$ iterations of the coloring phase all agents are properly colored. 
This leads to the following theorem.

\begin{theorem}\label{coloring_theorem}
    The partitioning by coloring stabilises in $O(\log^3 n)$ fragmented time.
\end{theorem}

We conclude with the proof of Theorem~\ref{fast-median}.
\begin{proof}[Proof of Theorem~\ref{fast-median}]
    The structure of the solution replicates the logic of a standard median computation protocol. Thus, the correctness of the solution follows from the correctness of the individual routines including leader election, majority computation and partitioning.
    
    Concerning the time complexity,
    leader election and majority computation are implemented in fragmented parallel time $O(\log n)$ whp, see Lemmas \ref{LE-lemma} and \ref{MP}. By Theorem \ref{coloring_theorem}, 
    each partitioning stage takes $O(\log^3 n)$ time whp, and with probability $\frac{1}{2}$ at most $\frac{3}{4}$ candidates remain in $C$. Thus, with high probability after at most $O(\log n)$ iterations of this routine set $C$ is reduced to a singleton containing the median.

    In conclusion, Fast-median protocol stabilises in $O(\log^4 n)$ fragmented parallel time.
\end{proof}

\section{Suitable Programming Environment}~\label{proglang}

It is a common inclination to articulate solutions in any computational model using pseudocode. Such representation enhances the readability and understanding of the proposed solution within the context of the main features of the underlying computational model. Subsequently, this aids in conducting rigorous mathematical analysis. 
%
Various pseudocodes have been explored in the past to address challenges in population protocols, encompassing simple protocols~\cite{DBLP:journals/corr/abs-1802-06872}, separation bounds~\cite{Czerner23,abs-2204-02115}, and leader-based computation~\cite{DBLP:journals/dc/AngluinAE08a}. The latter work forms the basis for our approach, which here focuses on the development of efficient parallel protocols.
Our objective is to champion selective population protocols, enabling the development of simpler and more structured efficient solutions presented at higher programming level. The primary reasons for advocating this approach stem from the partitioning of the state space. Each partition represents the local variables of an independent process, supported by conditional interactions that also facilitate independent interaction availability (zero) tests. This, in turn, eliminates the necessity for a global clocking mechanism through the application of event-based distributed computation.

It is commonly assumed that population protocols provide a computing environment operating on cardinality variables with certain additional restrictions. Each variable corresponds to the number of agents in a given state, and the sum of all variables does not exceed the size of the population. During an interaction, one can modify at most two variables.
As mentioned earlier, selective protocols go beyond traditional population protocols, introducing tools for more structured computing. This includes the sequential or simultaneous (in parallel) execution of independent processes supported by event-based (zero cardinality variable test) synchronization.
Thus, one of the primary challenges in using pseudocode lies in capturing the natural parallelism of population protocols, particularly in effectively representing transition function rules in the corresponding pseudocode.

Below, we present a pseudocode example illustrating the majority protocol \textbf{M-protocol} discussed in Section~\ref{Majority}. The pseudocode employs standard programming instructions, incorporating functions (representing independent population processes), as well as conditional (if) and parallel (while) statements to provide a control flow mechanism, all presented in a Python-like syntax.
The computation relies on two fundamental atomic operations: \texttt{Inc(X)}; \texttt{Dec(Y)}, which facilitate the movement of agents between different states through the incrementing and decrementing of cardinality variables $X$ and $Y$, equivalent with operation $X\mapsto Y$ in \cite{abs-2204-02115}.

\begin{tcolorbox}[left=4em]
\begin{lstlisting}
def first_wins(X,Y,Z): 
  while X$>0$ and Y$>0$ in parallel:  
    Dec(X); Dec(Y); Inc(Inc(Z));  
  if X$>0$: return true
  else: return false
  
def rewrite(X,Y): 
  while X$>0$ in parallel:
    Dec(X); Inc(Y);

def M-protocol():
  if first_wins($R$,$G$,$N$): // $R>0$
    rewrite($R\vee N$,$R^*$);
  else: // $G>0$
    rewrite($G\vee N$,$G^*$);
\end{lstlisting}
\end{tcolorbox}

Although no explicit leader states are used in the code, all computation threads 
can be mapped onto the corresponding rules of the majority protocol.
%
In particular, transition rules (1) and (3) are encoded in lines 1--5 as function {\tt first\_wins} (process based on effective pairng in color cancellation)
which returns value {\em true} when $X>Y,$ and {\em false} otherwise.
The remaining rules are encoded in function {\tt rewrite} (process based on the epidemic) with reference to 
rules (2) and (5) in line 13 and rules (4) and (6) in line 15.

The presented pseudocode illustrates the feasibility of a formally defined high-level programming language that can be compiled into selective population protocols, ensuring both efficiency and readability of the solution. Given that this represents the initial attempt, it is crucial to acknowledge that further exploration and development are required for additional details regarding the formalism, including the syntax and semantics of each instruction.

\section{Final Comment}~\label{final} 
We would like to postulate that the efficiency of selective protocols stand out when tackling problems that demand more extensive memory utilization and yield intricate outputs. Examples of such challenges include the {\em ranking problem}~\cite{BurmanCCDNSX21}, examined recently in the context of leader election in self-stabilizing protocols, and related sorting problem~\cite{DBLP:conf/stacs/GasieniecSS23} studied earlier in the constructors model. For these two problems, no efficient solutions based on a polynomial number of states are currently known in standard population protocols.

In fact, we would like to assert the following.

%
\begin{conjecture}
Any efficient solution to the sorting problem necessitates exponential state space in standard population protocols. 
\end{conjecture}

On the contrary, evidence presented in Section~\ref{ranking} of the Appendix demonstrates that selective protocols can efficiently solve sorting by ranking using much smaller number of states.
Specifically, we present transition rules of an efficient quick-sort-like, selective sorting by ranking. This algorithm has polynomial in $n$ state space, utilises $O(n)$ partitions and stabilises in time $O(\log^2 n)$.

\section{Acknowledgements}
This research was supported in part by the National Science Centre, Poland, under project number 2020/39/B/ST6/03288 (Adam Gańczorz, Tomasz Jurdziński, Grzegorz Stachowiak). 

This research was supported in part by the National Science Centre, Poland, under project number 2021/41/B/ST6/03691 (Jakub Kowalski).






\bibliographystyle{plain}
\bibliography{bibliography}

\begin{thebibliography}{10}

\bibitem{DBLP:conf/soda/AlistarhAEGR17}
D.~Alistarh, J.~Aspnes, D.~Eisenstat, R.~Gelashvili, and R.L. Rivest.
\newblock Time-space trade-offs in population protocols.
\newblock In {\em Proc. {SODA} 2017}, pages 2560--2579, 2017.

\bibitem{DBLP:conf/soda/AlistarhAG18}
D.~Alistarh, J.~Aspnes, and R.~Gelashvili.
\newblock Space-optimal majority in population protocols.
\newblock In {\em Proc. {SODA} 2018}, pages 2221--2239, 2018.

\bibitem{DBLP:conf/opodis/AlistarhGR21}
D.~Alistarh, R.~Gelashvili, and J.~Rybicki.
\newblock Fast graphical population protocols.
\newblock In Quentin Bramas, Vincent Gramoli, and Alessia Milani, editors, {\em
  25th International Conference on Principles of Distributed Systems, {OPODIS}
  2021, December 13-15, 2021, Strasbourg, France}, volume 217 of {\em LIPIcs},
  pages 14:1--14:18. Schloss Dagstuhl - Leibniz-Zentrum f{\"{u}}r Informatik,
  2021.

\bibitem{DBLP:conf/podc/AlistarhRV22}
D.~Alistarh, J.~Rybicki, and S.~Voitovych.
\newblock Near-optimal leader election in population protocols on graphs.
\newblock In {\em {PODC} '22: {ACM} Symposium on Principles of Distributed
  Computing, Salerno, Italy, July 25 - 29, 2022}, pages 246--256. {ACM}, 2022.

\bibitem{DBLP:conf/podc/AngluinADFP04}
D.~Angluin, J.~Aspnes, Z.~Diamadi, M.J. Fischer, and R.~Peralta.
\newblock Computation in networks of passively mobile finite-state sensors.
\newblock In {\em Proc. PODC 2004}, pages 290--299, 2004.

\bibitem{DBLP:journals/dc/AngluinAE08a}
D.~Angluin, J.~Aspnes, and D.~Eisenstat.
\newblock Fast computation by population protocols with a leader.
\newblock {\em Distributed Comput.}, 21(3):183--199, 2008.

\bibitem{BankhamerBBEHKK22}
G.~Bankhamer, P.~Berenbrink, F.~Biermeier, R.~Els{\"{a}}sser, H.~Hosseinpour,
  D.~Kaaser, and Peter Kling.
\newblock Population protocols for exact plurality consensus: How a small
  chance of failure helps to eliminate insignificant opinions.
\newblock In Alessia Milani and Philipp Woelfel, editors, {\em {PODC} '22:
  {ACM} Symposium on Principles of Distributed Computing, Salerno, Italy, July
  25 - 29, 2022}, pages 224--234. {ACM}, 2022.

\bibitem{BerenbrinkGK20}
P.~Berenbrink, G.~Giakkoupis, and P.~Kling.
\newblock Optimal time and space leader election in population protocols.
\newblock In Konstantin Makarychev, Yury Makarychev, Madhur Tulsiani, Gautam
  Kamath, and Julia Chuzhoy, editors, {\em Proceedings of the 52nd Annual {ACM}
  {SIGACT} Symposium on Theory of Computing, {STOC} 2020, Chicago, IL, USA,
  June 22-26, 2020}, pages 119--129. {ACM}, 2020.

\bibitem{BurmanCCDNSX21}
J.~Burman, H{-}L. Chen, H{-}P. Chen, D.~Doty, T.~Nowak, E.E. Severson, and
  C.~Xu.
\newblock Time-optimal self-stabilizing leader election in population
  protocols.
\newblock In Avery Miller, Keren Censor{-}Hillel, and Janne~H. Korhonen,
  editors, {\em {PODC} '21: {ACM} Symposium on Principles of Distributed
  Computing, Virtual Event, Italy, July 26-30, 2021}, pages 33--44. {ACM},
  2021.

\bibitem{DBLP:conf/podc/BurmanCCDNSX21}
J.~Burman, H{-}L Chen, H{-}P Chen, D.~Doty, T.~Nowak, E.E. Severson, and Ch.
  Xu.
\newblock Time-optimal self-stabilizing leader election in population
  protocols.
\newblock In {\em {PODC} '21: {ACM} Symposium on Principles of Distributed
  Computing, Virtual Event, Italy, July 26-30, 2021}, pages 33--44. {ACM},
  2021.

\bibitem{abs-2204-02115}
P.~Czerner.
\newblock Leaderless population protocols decide double-exponential thresholds.
\newblock {\em CoRR}, abs/2204.02115, 2022.

\bibitem{Czerner23}
P.~Czerner.
\newblock Brief announcement: Population protocols decide double-exponential
  thresholds.
\newblock In Rotem Oshman, Alexandre Nolin, Magn{\'{u}}s~M. Halld{\'{o}}rsson,
  and Alkida Balliu, editors, {\em Proceedings of the 2023 {ACM} Symposium on
  Principles of Distributed Computing, {PODC} 2023, Orlando, FL, USA, June
  19-23, 2023}, pages 28--31. {ACM}, 2023.

\bibitem{DBLP:journals/ipl/CzumajL23}
A.~Czumaj and A.~Lingas.
\newblock On parallel time in population protocols.
\newblock {\em Inf. Process. Lett.}, 179:106314, 2023.

\bibitem{DotyEGSUS21}
D.~Doty, M.~Eftekhari, L.~G{\k a}sieniec, E.E. Severson, P.~Uznanski, and
  G.~Stachowiak.
\newblock A time and space optimal stable population protocol solving exact
  majority.
\newblock In {\em 62nd {IEEE} Annual Symposium on Foundations of Computer
  Science, {FOCS} 2021, Denver, CO, USA, February 7-10, 2022}, pages
  1044--1055. {IEEE}, 2021.

\bibitem{GasieniecS21}
L.~G{\k a}sieniec and G.~Stachowiak.
\newblock Enhanced phase clocks, population protocols, and fast space optimal
  leader election.
\newblock {\em J. {ACM}}, 68(1):2:1--2:21, 2021.

\bibitem{DBLP:conf/stacs/GasieniecSS23}
L.~G\k{a}sieniec, P.G. Spirakis, and G.~Stachowiak.
\newblock New clocks, optimal line formation and self-replication population
  protocols.
\newblock In {\em {STACS} 2023, March 7-9, 2023, Hamburg, Germany}, volume 254
  of {\em LIPIcs}, pages 33:1--33:22, 2023.

\bibitem{DBLP:conf/icalp/GuerraouiR09}
R.~Guerraoui and E.~Ruppert.
\newblock Names trump malice: Tiny mobile agents can tolerate byzantine
  failures.
\newblock In {\em {ICALP} 2009, Rhodes, Greece, July 5-12, 2009, Part {II}},
  volume 5556 of {\em Lecture Notes in Computer Science}, pages 484--495.
  Springer, 2009.

\bibitem{geometric}
S.~Janson.
\newblock Tail bounds for sums of geometric and exponential variables.
\newblock {\em Statistics and Probability Letters}, 135(1):1--6, 2018.

\bibitem{Ken98}
J.~Kennedy.
\newblock Thinking is social: Experiments with the adaptive culture model.
\newblock {\em Journal of Conflict Resolution}, 42(1):56--76, 1998.

\bibitem{Klein00}
J.M. Kleinberg.
\newblock The small-world phenomenon: an algorithmic perspective.
\newblock In {\em Proceedings of the Thirty-Second Annual {ACM} Symposium on
  Theory of Computing, May 21-23, 2000, Portland, OR, {USA}}, pages 163--170.
  {ACM}, 2000.

\bibitem{DBLP:conf/podc/KosowskiU18}
A.~Kosowski and P.~Uzna\'nski.
\newblock Brief announcement: Population protocols are fast.
\newblock In {\em Proceedings of {PODC} 2018, Egham, United Kingdom, July
  23-27, 2018}, pages 475--477. {ACM}, 2018.

\bibitem{DBLP:journals/corr/abs-1802-06872}
A.~Kosowski and P.~Uzna\'nski.
\newblock Population protocols are fast.
\newblock {\em CoRR}, abs/1802.06872, 2018.

\bibitem{MS17}
O.~Michail and P.G. Spirakis.
\newblock Network constructors: {A} model for programmable matter.
\newblock In {\em {SOFSEM} 2017, Limerick, Ireland, January 16-20, 2017,
  Proceedings}, volume 10139 of {\em Lecture Notes in Computer Science}, pages
  15--34. Springer, 2017.

\bibitem{minsky_67}
M.L. Minsky.
\newblock {\em Computation: Finite and Infinite Machines}.
\newblock Prentice-Hall Series in Automatic Computation. Prentice-Hall, 1967.

\bibitem{Murty}
S.~Murty.
\newblock Answer to: What is the expected number of comparisons of a sorting
  algorithm that chooses a random pair of elements and swaps if out of order,
  until sorted?
\newblock Quora, April 2019.

\bibitem{NearyW09}
T.~Neary and D.~Woods.
\newblock Four small universal turing machines.
\newblock {\em Fundam. Informaticae}, 91(1):123--144, 2009.

\bibitem{SoloveichikCWB08}
D.~Soloveichik, M.~Cook, E.~Winfree, and J.~Bruck.
\newblock Computation with finite stochastic chemical reaction networks.
\newblock {\em Nat. Comput.}, 7(4):615--633, 2008.

\bibitem{Young}
N.~Young.
\newblock Answer to: Expected number of random comparisons needed to sort a
  list.
\newblock Theoretical Computer Science Stack Exchange, October 2022.

\end{thebibliography}

\newpage

\appendix

\section{Appendix}
\pagenumbering{roman}

In this document, we offer a detailed description of aforementioned fast multiplication protocol, the missing proof for Theorem~\ref{MP}, and an outline for an efficient selective ranking protocol. 

\subsection{Fast multiplication}
\label{FM}

In this section we propose a fast multiplication protocol stabilising in $O(\log n)$ fragmented time.
Recall from Section~\ref{Beyond} that the input of multiplication is formed of two sets of agents $X$ and $Y,$ where $X$ contains all agents in state $x$ and $Y$ all agents in state $y.$ 
The main task is to create a new set $Z$ of agents in state $z,$ s.t., $|X|\cdot|Y|=|Z|.$

The multiplication protocol operates in $r=\lceil\log |Y|\rceil$ rounds which correspond to $r$ bits $b_0,b_1,\dots,b_{r-1}$ of binary expansion of the cardinality $|Y|.$
Each round is split into two stages.
During {\bf Stage 1} the next bit $b_i$ of this expansion is computed starting from the least significant bit $b_0.$ 
This is done by testing parity of the current size of set $Y,$ which is reduced by half to the size $|Y| \ div\ 2$ during each round.
In {\bf Stage 2} set $Z$ is increased by $b_i\cdot 2^i|X|$ agents. Thus, after $j$ rounds the size of $Z$ is $\sum_{i=0}^j b_i\cdot 2^i|X|$, and after $r=\lceil\log |Y|\rceil$ rounds $\sum_{i=0}^{r} b_i\cdot 2^i|X|=|X|\cdot|Y|$.

\begin{figure*}[hb]
\begin{small}
\begin{center}
\mbox{
\begin{tcolorbox}[width=7cm]
{\bf Stage 1: $|Y|\leftarrow |Y|/2$, remainder $b$ }
\begin{alignat*}{5}
\transitionnull{1a}{L_{in}}{\mathcal{G}_Y|null}{\ L_{out}}{}
\transition{1b}{L_{in}}{\mathcal{G}_Y|Y}{\ L}{A}
\transition{2a}{A}{\mathcal{G}_Y|Y}{\ A}{A}
\transitionnull{2b}{A}{\mathcal{G}_Y|null}{\ A'}{}
\transition{3a}{A'}{\mathcal{G}_A|A/A'}{\ B}{free}
\transitionnull{3b}{A'}{\mathcal{G}_A|null}{\ A''}{}
\transition{4a}{L}{\mathcal{G}_A|A''}{\ L^1}{free}
\transitionnull{4b}{L}{\mathcal{G}_A|null}{\ L^0}{}
\transitionnull{5}{B}{\mathcal{G}_A|null}{\ Y}{}
\transitionnull{6a}{L^0}{\mathcal{G}_B|null}{\ L_{in}^0}{}
\transitionnull{6b}{L^1}{\mathcal{G}_B|null}{\ L_{in}^1}{}
\end{alignat*}
\vspace{-3.2em}
%
\end{tcolorbox}
\hspace{10pt}
\begin{tcolorbox}[width=7cm]
{\bf Stage 2: $|Z|\leftarrow |Z|+b |X|;|X|\leftarrow 2|X|$
}
\begin{alignat*}{5}
\transition{1a}{L_{in}^0}{\mathcal{G}_X|X}{\ L'}{C^0}
\transition{1b}{L_{in}^1}{\mathcal{G}_X|X}{\ L'}{C^1}
\transition{2a}{C^0}{\mathcal{G}_X|X}{\ C^0}{C^0}
\transition{2b}{C^1}{\mathcal{G}_X|X}{\ C^1}{C^1}
\transitionnull{2c}{C^0}{\mathcal{G}_X|null}{\ C_\star^0}{}
\transitionnull{2d}{C^1}{\mathcal{G}_X|null}{\ C_\star^1}{}
\transition{3}{C_\star^1}{\mathcal{G}_f|free}{\ C_\star^0}{Z}
\transition{4}{C_\star^0}{\mathcal{G}_f|free}{\ D}{D}
\transitionnull{5a}{D}{\mathcal{G}_C|null}{\ X}{}
\transitionnull{5b}{L'}{\mathcal{G}_C|null}{\ L''}{}
\transitionnull{6}{L''}{\mathcal{G}_D|null}{\ L_{in}}{}
\end{alignat*}
\vspace{-3.5em}
\end{tcolorbox}
}
\end{center}
\end{small}
\caption{Transition rules for the  fast multiplication protocol.}
\label{code:FastMulti}
\end{figure*}

The multiplication protocol utilises the leader computed by LE-protocol from section~\ref{LE}.
The state space $S$ is partitioned into  
$\mathcal{G}_L=\{L_{in},L_{out},L,L^0,L^0_{in},L^1,L^1_{in}\}$,
$\mathcal{G}_A=\{A,A',A''\},$ 
$\mathcal{G}_B=\{B\},$
$\mathcal{G}_C=\{C^0, C^0_\star, C^1, C^1_\star\},$
$\mathcal{G}_D=\{D\},$
$\mathcal{G}_X=\{X\},$
$\mathcal{G}_Y=\{Y\},$
$\mathcal{G}_Z=\{Z\},$ and
$\mathcal{G}_f=\{free\}.$
The transition rules of the two stages are listed in Fig~\ref{code:FastMulti}.

\begin{theorem}\label{multi}
    The multiplication protocol stabilises on a subpopulation $Z$ for which $|Z|=|X|\cdot |Y|$
    where $|X|$ and $|Y|$ are initial sizes of subpopulations $X$ and $Y$.
    The stabilisation time is $O(\log |Y|\cdot\log n)$ whp.
\end{theorem}

\begin{proof}
    In {Stage 1} the unique leader can be in one of seven states including: $L_{in}$ in which the leader enters each round of the protocol, $L_{out}$ indicating that the multiplication protocol is finished, $L$ indicating neutral leadership in the first stage, $L^0$ and $L^1$ indicating that the leader knows the parity bit 0 or 1 respectively, $L^0_{in}$ and $L^1_{in}$ indicating that the leader is ready to start the second stage. State $Y$ indicates membership in $Y$ and states $A,A',A''$ and $B$ are used to denote agents temporarily hosting agent from set $Y$.

{Stage 1} assumes the input $L_{in},Y,$ where $L_{in}$ is the leader in the input state, and $Y$ represents the set of agents in state $Y.$
The output is defined by  the leader in state $L^{b}_{in}$ (ready for the second stage) where $b$ is the parity bit of $|Y|$, and reduced in size $Y$ with the cardinality $|Y|=\lfloor\frac{|Y|}{2}\rfloor$.
Alternatively, the leader adopts state $L_{out},$ when 
initially 
$|Y|=0,$ which concludes the multiplication process.
Inside {Stage 1} rule $(1a)$ 
implements
In conclusion, as all actions described above refer to a fixed number of independent epidemic and emptiness test processes, {Stage 2} also operates in $O(\log n)$ fragmented parallel time.

The emptiness test on $Y,$ and rule $(1b)$ initiates the epidemic of state $A$ among agents with the input state $Y.$ It also creates the internal leader $L$ which remains idle until creation of state $A''$ or complete depletion of agents in group $G_A$. Further, rule $(2a)$ supports the epidemic of $A$s until rule $(2b)$ is triggered which concludes this epidemic with creation of state $A'.$ From now on, each agent in state $A$ changes its state inside  group $G_A$ to $A'$ by rule $(2b),$ or becomes a free agent by application of rule $(3a).$ Note that rule $(3a)$ removes two agents from group $G_A$ moving one agent to group $G_B$ and the other to $G_f.$ This guarantees that while the maximum size of $G_A$ 
is 
the same as 
the initial $|Y|,$
the maximum size of $G_B$ 
corresponds to $\lfloor \frac{|Y|}{2}\rfloor.$
On the conclusion of this process, either state $A''$ is created by rule $(3b)$ meaning that an odd element was left in $G_A,$ i.e., $|Y|$ is odd. This is confirmed by rule $(4a)$ which gives the leader state $L^1.$ Otherwise, rule $(4b)$ recognises that there are no agents left in $G_A,$ which is 
reflected
by state $L^0.$ 
In addition, utilising emptiness test on $G_A$ rule $(5)$ relocates all agents from $G_B$ 
to $G_Y.$ Finally, when this relocation is concluded the leader gets ready for {\bf Stage 2} with the help of rule $(6a)$ or $(6b).$
In conclusion, as all actions described above refer to a fixed number of independent epidemic and emptiness test processes, {\bf Stage 2} operates in $O(\log n)$ fragmented parallel time.

In {\bf Stage 2} the unique leader can be in one of five states including the two initial states: $L^0_{in}$ and $L^1_{in}$ indicating whether the computed parity bit in {\bf Stage 1} is 0 or 1, respectively.
The agents in state $X$ are temporarily relocated to group $G_C$ where they initially reside either in state $C^0$ or $C^1.$
This process is initiated by the leader in the relevant state by rule $(1a)$ or $(1b)$, and further continued in the form of epidemic 
performed 
by rule $(2a)$ or $(2b).$
The transfer to group $G_C$ halts when $G_X$ becomes empty. This is witnessed by arrival of state $C^0_{\star}$ or $C^1_{\star}$ during emptiness test utilised in rule $(2c)$ or $(2d)$ respectively.
Each agent in state $C^1_{\star}$ adds one agent to set $Z$ before adopting state $C^0_{\star},$ see rule $(3),$
which implements the operation $|Z|\leftarrow |Z|+|X|$ when the parity bit is $1$.
Rule $(4)$ creates group $G_D$ which becomes of double the size of $G_C.$ And when $G_C$ becomes empty, i.e., group $G_D$ is fully formed, rule $(5a)$ moves all agents from $G_D$ to $G_X$ and rule $(5b)$ instructs the leader to test the emptiness of $G_C.$ And when $G_D$ becomes empty, the leader adopts state $L_{in},$ rule $(6),$ making it ready for the next iteration beginning from the execution of {\bf Stage 1}.  
In conclusion, as all actions described above refer to a fixed number of independent epidemic and emptiness test processes, {\bf Stage 2} operates in $O(\log n)$ fragmented parallel time.
\end{proof}


    


\subsection{Quick-sort like selective ranking protocol}\label{ranking}

Recall that in the \textit{ranking problem}, the primary challenge is to assign unique numbers from the range 1 to $n$ to all agents. Furthermore, in \textit{sorting by ranking}, the order of keys of agents must align consistently with their assigned ranking.

We present on the next page, see Figure~\ref{code:RankAlgo},  pseudocode for a ranking selective protocol that illustrates a parallel solution based on multiple computation threads. The protocol adopts a divide-and-conquer strategy, where initially, all agents belong to the same group with index (offset) $r=1$. After a random pivot $P'$ is chosen (Phase 0), the remaining agents in this group are split into two subpopulations (Phase 1). One subpopulation $B$ (below) is formed by all agents with keys smaller than the pivot's key, while the other, $A$ (above), contains all agents with keys greater than the pivot's key.

The target rank $rank(P')=r+|B|$ of the pivot is then computed based on the index $r$ of the group and the cardinality of $B$ (Phase 2). Finally, a new group with offset $rank(P')+1$ is formed (Phase 3), where all agents from subpopulation $A$ are transferred. After this transition, agents in both groups are prepared for the new recursive round of pivot choosing and splitting (Phase 4), continuing in parallel until all groups are reduced to singletons.


\newpage
\thispagestyle{empty}
\begin{figure*}
\vspace*{-5pt}
\newcommand{\GR}{\mathcal{G}_r}
\newcommand{\GPR}{\mathcal{G}'_r}
\newcommand{\GBR}{\mathcal{G}''_r}
\begin{small}
\begin{center}
\begin{tcolorbox}
In groups $\GR$, $\GPR$, $\GBR$, the index $r$ is the minimal rank of an element within a group.\\
Agent states are encoded as triplets $(t, k, v)$, where $t$ is the type of the agent, used to differentiate its role in the computation process, $k$ is its key, constant for each agent, and $v$ encodes additional value required for computations, its meaning depends on the agent type and phase of the protocol.\\\\
If, inside a transition rule, an element from state's tuple remains unchanged, we mark such situation using $*$. When the value does not matter for further computations $\_$ is used instead.\\
If it is not clear from the context of the rule to which group an agent of type $T$ belongs, we write it explicitly as $T[\mathcal{G}]$.\\\\
Phase 0: Choosing the initial pivot for the entire population of agents. All agents start as $(L,k,\_)\in\mathcal{G}'_1$, where $k$ is each agent's key, know to them and used for comparion with other keys.\\
$N \in\mathcal{G}_1$, $L, P \in\mathcal{G}'_1$
\noindent\begin{alignat*}{5}
\transition{0a,\text{ leader election}}{(L,*,*)}{\mathcal{G}'_1|(L,*,*)}{(L,*,*)}{(N,*,\_)\qquad\qquad\qquad\qquad\qquad}
\transitionnull{0b, \text{leader becomes pivot for the next phase}}{(L,k,*)}{\mathcal{G}'_1|null}{(P,k,k)}{}
\end{alignat*}

\vspace*{-0.4cm}
Phase 1: For every group $\GR$ initially containing only neutral agents $N$, and $\GPR$ with a single pivot $P$, the $\GR$ members are transferred to $\GPR$, divided into being above ($A$) or below ($B$) the pivot, and learn the key of the pivot (stored in the state's third value $v$).\\
$N\in\GR$, $P, P', A, A', B, B' \in\mathcal{G}'_r$
\noindent\begin{alignat*}{5}
\transition{1a, \text{if } k<l \text{: pivot sets neutral as above}}{(P,k,k)}{\GR|(N,l,\_)}{(P,k,k)}{(A,l,k)\qquad\qquad\qquad\qquad\qquad\qquad}
\transition{1a, \text{if } k>l \text{: pivot sets neutral as below}}{(P,k,k)}{\GR|(N,l,\_)}{(P,k,k)}{(B,l,k)}
\transition{1b, \text{if } k<l \text{: above/below sets neutral as above}}{(A/B,*,k)}{\GR|(N,l,\_)}{(A/B,*,k)}{(A,l,k)}
\transition{1b, \text{if } k>l \text{: above/below sets neutral as below}}{(A/B,*,k)}{\GR|(N,l,\_)}{(A/B,*,k)}{(B,l,k)}
\transitionnull{1c, \text{pivot/above rewritten for the next phase}}{(P/A,*,*)}{\GR|null}{(P'/A',*,0)}{}
\transitionnull{1d, \text{below rewritten for the next phase}}{(B,*,*)}{\GR|null}{(B',*,1)}{}
\end{alignat*}

\vspace*{-0.4cm}
Phase 2: Establishing the rank of the pivot $P'$ within the group $\GPR$, which is equal to $|B'|+1$. The state value $v$ is used as a storage for partial results for the rank computation.\\
$P', A',B'\in\mathcal{G}'_r$, $A'',B'',P''\in\mathcal{G}''_r$

\noindent\begin{alignat*}{4}
&\omit\rlap{($2a$, $A,A'$ is rewritten for next phase as $A''$)}  \\
\transitionreltwo{}{(A'/B'/P',*,*)}{\mathcal{G}'_r|(A/A',*,*)}{(A'/B'/P',*,*)}{(A'',*,\_)}
&\omit\rlap{($2b$, counting $B'$, agents already counted and rewritten to next phase)} && \\
\transitionreltwo{}{(P'/B',*,x)}{\mathcal{G}'_r|(B',*,y)}{(P'/B',*,x+y)}{(B'',*,\_)\qquad\qquad\qquad\qquad\qquad\qquad\qquad\qquad\qquad}
&\omit\rlap{($2c$, pivot knows the group size, rewritten for the next phase)} && \\
\transitionrelnulltwo{}{(P',*,x)}{\mathcal{G}'_r|null}{(P'',*,x)}{}
\end{alignat*}

\vspace*{-0.4cm}
Phase 3: Resetting agent, above and below lose their type and become neutrals in proper groups, either $\GR$ or $\mathcal{G}_{r+x+1}$. Additionally, above agents learn the pivot's rank $r+x$.
State value is used for storing index of the agent's old group or its rank within it, depending on particular need. \\
$B'',P''\in\mathcal{G}''_r$, $N'\in \mathcal{G}'_{r+x+1}$
\noindent\begin{alignat*}{5}
\transitionnull{3a,\text{below become neutral in group $\GR$ }}{(B'',*,*)}{\GPR|null}{(N[\GR],*,\_)}{}
\transition{3b, \text{above become neutral $N'\in\mathcal{G}'_{r+x+1}$}}{(P'',*,x)}{\mathcal{G}''_r|(A'',*,*)}{(P'',*,x)}{(N',*,r)}
\transition{3b, \text{above become neutral $N'\in\mathcal{G}'_{r+x+1}$}}{(N',*,r)}{\mathcal{G}''_r|(A'',*,*)}{(N',*,r)}{(N',*,r)}
\transitionnull{3c,\text{$N'$ become standard neutral in $\mathcal{G}_{r+x+1}$ }}{(N',*,r)}{\mathcal{G}''_r|null}{(N[\mathcal{G}_{r+x+1}],*,\_)}{}
\transitionnull{3d,\text{pivot moves to group corresponding to its rank}}{(P'',*,x)}{\mathcal{G}''_r|null}{(P'''[\mathcal{G}_{r+x}],*,r)}{}
\end{alignat*}

\vspace*{-0.4cm}
Phase 4: Randomizing new pivots for groups $\GPR$ (containing below agents) and $\mathcal{G}'_{r+x+1}$ (containing above agents). Pivot's values are set to their keys. Old pivot becomes a sole new agent in group $\mathcal{G}_{r+x}$ with type $R$. Computations proceed to Phase 1.\\
$P''', P'''', P''''', R\in\mathcal{G}_{r+x}$
\noindent\begin{alignat*}{5}
\transitionnull{4a,\text{all $N'$ rewritten to $N$}}{(P''',*,r)}{\mathcal{G}'_{r+x+1}|null}{(P'''',*,r)}{}
\transition{4b, \text{indication of the new pivot in $\mathcal{G}_{r+x+1}$}}{(P'''',*,r)}{\mathcal{G}_{r+x+1}|(N,k,*)}{(P''''',*,r)}{(P[\mathcal{G}'_{r+x+1}],k,k)}
\transitionnull{4b,\text{in case no pivot candidates in $\mathcal{G}_{r+x+1}$}}{(P'''',*,r)}{\mathcal{G}_{r+x+1}|null}{(P''''',*,r)}{}
\transition{4c, \text{indication of the new pivot in $\GR$}}{(P''''',*,r)}{\mathcal{G}_r|(N,k,*)}{(R,*,*)}{(P[\GPR],k,k)}
\transitionnull{4c,\text{in case no pivot candidates in $\GR$}}{(P''''',*,r)}{\mathcal{G}_{r}|null}{(R,*,*)}{}
\end{alignat*}
\vspace{-28pt}
\end{tcolorbox}
\end{center}
\end{small}
\vspace*{-0.1cm}
\caption{Transition rules of quick-sort like selective ranking protocol}
\label{code:RankAlgo}
\end{figure*}



    

\end{document}